\documentclass[preprint, 10pt]{elsarticle}
\usepackage{lineno}
\usepackage{hyperref}
\hypersetup{hidelinks}

\usepackage{color}

\usepackage[left=2cm,right=2cm,top=2.5cm,bottom=2.5cm]{geometry}

\usepackage{booktabs}
\usepackage{graphicx}
\usepackage{subfigure}
\usepackage{epstopdf}
\usepackage{picinpar}
\usepackage{amsmath}
\usepackage{amsfonts}
\usepackage[mathscr]{euscript}
\usepackage{calligra}
\usepackage{array}
\DeclareMathAlphabet{\mathpzc}{OT1}{pzc}{m}{it}

\newcommand{\eg}{\textit{e}.\textit{g}.}

\usepackage[linesnumbered,ruled,lined]{algorithm2e}
\usepackage{amssymb}

\newtheorem{theorem}{Theorem}

\newdefinition{assumption}{Assumption}

\newdefinition{definition}{Definition}
\newdefinition{remark}{Remark}
\newproof{proof}{Proof}

\journal{}

\biboptions{authoryear,compress}
\let\OLDthebibliography\thebibliography
\renewcommand\thebibliography[1]{
  \OLDthebibliography{#1}
  \setlength{\parskip}{0pt}
  \setlength{\itemsep}{0pt plus 0.3ex}
}

\begin{document}

\begin{frontmatter}

\title{Multi-lane Unsignalized Intersection Cooperation with Flexible Lane Direction\\based on Multi-vehicle Formation Control}

\author[thu]{Mengchi Cai}
\ead{cmc18@mails.tsinghua.edu.cn}

\author[thu]{Qing Xu\corref{cor1}}
\ead{qingxu@tsinghua.edu.cn}

\author[thu]{Chaoyi Chen}
\ead{chency19@mails.tsinghua.edu.cn}

\author[thu]{Jiawei Wang}
\ead{wang-jw18@mails.tsinghua.edu.cn}

\author[thu]{Keqiang Li}
\ead{likq@tsinghua.edu.cn}

\author[thu]{Jianqiang Wang}
\ead{wjqlws@tsinghua.edu.cn}

\author[intel]{Xiangbin Wu}
\ead{xiangbin.wu@intel.com}

\address[thu]{School of Vehicle and Mobility, Tsinghua University, Beijing, China}
\address[intel]{Intel Lab China, Beijing, China}

\cortext[cor1]{Corresponding author}

\begin{abstract}
Unsignalized intersection cooperation of connected and automated vehicles (CAVs) is able to eliminate green time loss of signalized intersections and improve traffic efficiency. Most of the existing research on unsignalized intersection cooperation considers fixed lane direction, where only specific turning behavior of vehicles is allowed on each lane. Given that traffic volume and the proportion of vehicles with different turning expectation may change with time, fixed lane direction may lead to inefficiency at intersections. This paper proposes a multi-lane unsignalized intersection cooperation method that considers flexible lane direction. The two-dimensional distribution of vehicles is calculated and vehicles that are not in conflict are scheduled to pass the intersection simultaneously. The formation reconfiguration method is utilized to achieve collision-free longitudinal and lateral position adjustment of vehicles. Simulations are conducted at different input traffic volumes and turning proportion of incoming vehicles, and the results indicate that our method outperformances the fixed-lane-direction unsignalized cooperation method and the signalized method.

\end{abstract}


\end{frontmatter}

%
\section{Introduction}
\label{intro}
%

Intersection is one of the most common traffic scenarios where congestion happens. Numerous intelligent intersection management methods are proposed to reduce congestion and improve traffic efficiency~\citep{bian2019cooperation,chen2021mixed,ge2021real}. Among the existing research about intersection management, signalized intersections and unsignalized intersections are the two main scenarios. At signalized intersections, researchers optimize traffic signal timing, control vehicles to adapt to the signal timing, and combine the two ideas to achieve cooperative optimization of both signal timing and vehicle control~\citep{mirchandani2001real,xu2018cooperative}. However, the stop-and-go behavior of vehicles at signalized intersections can't be avoided due to the natural property of traffic signals, which may bring about green time loss and result in inefficiency for traffic management. Hence, researchers have started to focus on methods for unsignalized intersection cooperation.

In the field of unsignalized intersection management, reservation methods, planning methods, and optimization methods are three main categories. Reservation methods centralizedly organize trajectories of vehicles where each vehicle reserves a block of space-time in the intersection and the centralized controller manages the reservations according to the “First Come, First Served” policy~\citep{dresner2008multiagent,lee2012development,huang2012assessing}. Reservation is a simple way to achieve multi-vehicle collision-free passing without traffic signal, but may not fully utilize the cooperation potential of vehicles and limit the improvement of traffic efficiency. Planning methods usually plan vehicles' spatiotemporal trajectories to resolve conflict and avoid collision between vehicles. Passing order is the main issue that most of the planning methods consider, and some algorithms are designed to calculate passing order of vehicles, \eg, depth-first searching methods~\citep{xu2018distributed}, Monte Carlo Tree Searching methods~\citep{xu2019cooperative}, and other graphical methods~\citep{li2006cooperative,chen2021conflict,chen2021graph}. Optimization methods centralizedly model the behavior of vehicles and solve the optimization problem with collision-avoidance constraints. Some methods have been proposed to solve the centralized optimization problem, \eg, dynamic programming methods~\citep{yan2009autonomous}, model predictive control (MPC) methods~\citep{kamal2014vehicle}, and other optimization methods~\citep{wu2012cooperative,lee2013sustainability}.

Among the above intersection management methods, fixed lane direction is commonly considered, where only specific turning is allowed on each lane. In fact, fixed lane direction is a legacy from the signalized intersections to guide vehicles with different turning expectation to distribute on different lanes and avoid collision when passing the intersection. However, considering the fluctuation of traffic input where proportion of vehicles with different turning expectation may change with time, the capacity of some lanes can't be fully utilized in fixed-lane-direction intersections. Thus, removing the constraint of fixed lane direction and making full use of intersection capacity can further improve traffic efficiency with help of CAVs. Some research has focused on dynamic lane use and lane-allocation-free intersection management. Dynamic lane use methods usually adjust the lane-use configuration of intersections according to the changing traffic inputs. Some research combines dynamic lane use and signal control to improve traffic efficiency at intersections~\citep{alhajyaseen2017integration,wu2012simulation}. Some other research extends the idea of dynamic lane use and make some lanes reversable (allowing reverse traffic flow) to balance traffic flow demand at different time~\citep{lu2018bi,phan2019space}. A cooperation method for lane-allocation-free intersections is proposed in~\cite{yu2019corridor} which optimizes vehicles' trajectories by mixed-integer linear programming. A lane-change-free cooperation method is proposed in~\cite{he2018erasing} for future road intersections, where vehicles don't change lane and adjust their longitudinal position to pass the intersection without collision. Other methods considering dynamic lane use include~\cite{zhang2012dynamicc,zhao2013two}.

Most of the existing intersection management methods ignore or simplify the lane changing behavior of vehicles, \eg, lane changing is assumed to have finished before vehicles enter the intersections in~\cite{chen2021conflict,chen2021graph}, and is assumed to finish instantly in~\cite{yu2019corridor}. Success rate of lane changing has a great impact on the efficiency of intersections, and unsuccessful and inefficient lane changing is one of the main reasons that cause traffic congestion~\citep{cai2021formationa}. Some cooperative lane changing methods have been proposed in multi-lane scenarios, \eg, formation control methods, vehicle sorting methods, and priority-based methods.

Multi-lane formation control methods consider vehicles in a certain range as a group, and cooperatively control them to achieve formation forming, structure reconfiguration, and vehicle joining/leaving management~\citep{marjovi2015distributed,cai2019multi,zheng2021distance}. Both formation control and vehicle sorting methods have focused on the problem of collision avoidance when a group of vehicles perform cooperative lane changing~\citep{wu2021cooperative,cao2021platoon,cai2021formationb}. Vehicles' preference on specific lanes is considered in~\cite{cai2021formationc}, which is a key issue in intersection management because vehicles need to change to the lanes where their desired turning is allowed. A vehicle regrouping method based on multi-vehicle formation control is proposed in~\cite{xu2021coordinated} to guide vehicles to change lanes in a multi-lane intersection. Some priority-based methods are combined with heuristic searching algorithms to plan trajectories for vehicles considering lane changing at ramps and lane-drop bottlenecks~\citep{xu2020bi}. The above methods solve cooperative lane changing problems in multi-lane scenarios, and have the potential to be applied to multi-lane intersections.

This paper focuses on the problem of multi-lane unsignalized intersection cooperation with flexible lane direction. The main contributions of this paper include that:
\begin{enumerate}[(1)]
\item A two-dimensional distribution calculation method is proposed to calculate both longitudinal (passing order) and lateral (lane occupation) position for vehicles at intersections. Conflict relationship between different traffic movement is defined and available position is calculated for each vehicle considering achievability. The proposed method can be applied to multi-lane unsignalized intersections, either with flexible lane direction or with fixed lane direction. 
\item The formation reconfiguration method using relative motion planning and multi-stage vehicle control is utilized for vehicles to achieve the calculated distribution. Collision-free motion of vehicles is planned in the relative coordinate system. Multi-stage trajectories are planned for vehicles and trajectory following with spatiotemporal constraints is modeled and solved as an optimal control problem. Compared with existing research that ignores or simplify lane changing behavior of vehicles, this method enables efficient lane changing at intersections.
\item The performance of the proposed method is verified in simulations at different input traffic volumes and different input turning proportion of vehicles. Simulations show that removing lane direction constraints can significantly imporve passing capacity of intersections because more vehicles are allowed to pass the conflict zone of intersections simultaneously. The results indicate that the efficiency of the proposed method outperforms both the unsignalized cooperation method with fixed lane direction and the signalized method.
\end{enumerate}

The rest of this paper is organized as follows. Section~\ref{frame} introduces the framework of  multi-lane unsignalized intersection cooperation method. Section~\ref{conf} and Section~\ref{forma} propose the distribution calculation method and the formation reconfiguration method respectively. Section~\ref{simu} carries out simulations and Section~\ref{conc} gives the conclusion.

%
\section{Framework of unsignalized intersection management}
\label{frame}
%

In this section, we introduce the scenarios studied in this paper and propose the framework for multi-lane unsignalized intersection management with flexible lane direction.

\subsection{Scenarios}
\label{sce}

The scenario studied in this paper is the multi-lane unsignalized intersection with flexible lane direction. An example of three-lane intersection where all incoming lanes allow variable turning is shown in Fig.~\ref{scenarioFig}. A vehicle is allowed to turning right, going straight, and turning left on every incoming lane. The area where four arms of intersections overlap is conflict zone, shown as the area inside the dashed lines in Fig.~\ref{scenarioFig}. The conflict zone is also the area inside the stop line of intersections. 

The incoming lanes are marked in gray. Lanes are indexed by integers from right to left beginning with one, from the perspective of vehicles. In order to simplify and regulate routes of vehicles, the assumption is made that vehicles can only travel to the lane with the same index when passing the conflict zone of an intersection, and the number of lanes of each arm is the same.

\begin{figure}
\begin{center}
    \includegraphics[width=0.35\linewidth]{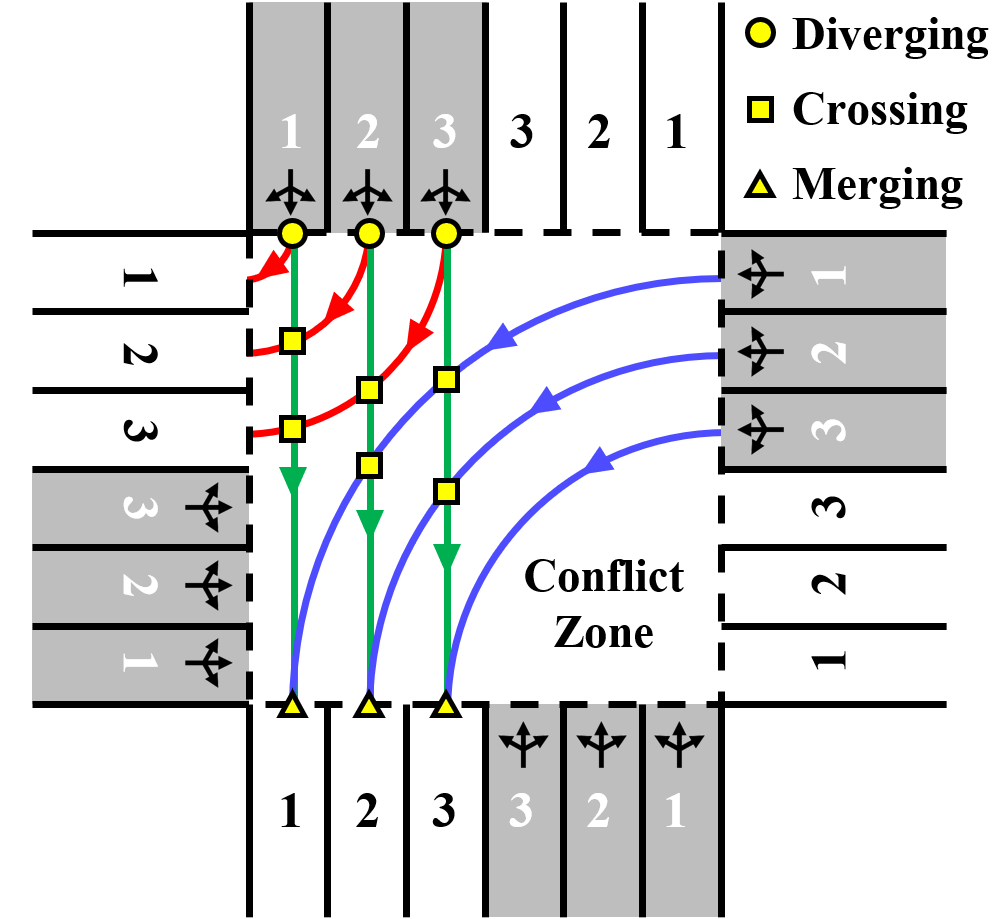}
    \caption{A three-lane intersection where variable turning is allowed for all the incoming lanes.}
    \label{scenarioFig}
\end{center}
\end{figure}

There are three types of traffic movements from one incoming arm, including right-turning movement, straight-going movement, and left-turning movement, and twelve movements in total from all the four arms. Three attributions are defined for each vehicle to describe their movement: $r$, $t$ and $l$. $r$ represents the arm where the vehicle comes from, and is set to 1, 2, 3, and 4 if the vehicle is on the north, east, south, and west arm respectively. $t$ represents the turning behavior, and is set to 1, 2, and 3 if the vehicle is going to turn right, go straight, and turn left respectively. $l$ represents the lane that the vehicle is occupying, and is set equal to the index of the lane. Then, in order to differentiate traffic movements of different vehicle, the movement variable $M$ of a vehicle is calculated as:
\begin{eqnarray}
\label{movement}
M=3n^\text{L}(r-1)+n^\text{L}(t-1)+l,
\end{eqnarray}
where $n^\text{L}$ represents the number of lanes on each arm, and there are $12n^\text{L}$ different movements in total. The definition of movement of vehicles in (\ref{movement}) will be used to check conflict between vehicles later in Section~\ref{conflict}.

\subsection{Framework}
\label{framework}

In order to control vehicles to safely pass an intersection without collision with others, we need to determine two issues for every single vehicle: which lane to occupy, and when to pass the conflict zone. Assuming that vehicles pass the conflict zone with a constant speed~\citep{xu2018distributed,chen2021conflict}, the issue of determining when to pass the conflict zone is equal to determining the relative longitudinal desired position of vehicles. Then the key of controlling vehicles to safely pass an intersection becomes: calculating longitudinal and lateral relative position of vehicles, guiding them to drive to their desired position, and control them to pass the conflict zone with a constant and safe speed.

For realizing the above steps, we propose a framework for multi-lane unsignalized intersection management with flexible lane direction. The framework consists of three parts: distribution calculation, formation reconfiguration, and vehicle control. 

\begin{enumerate}[(1)]
\item Distribution calculation. Both the longitudinal (passing order and following distance) and lateral (lane occupation) distribution of vehicles are calculated in this part, considering achievability (whether a vehicle is able to catch up with a target position before reaching the stop line). The calculated distribution guarantees that if all vehicles travel with the same constant speed, they can pass the conflict zone without collision.
\item Formation reconfiguration. In this part, all the vehicles driving on the same incoming arm are considered as a formation, and paths for them to travel to their desired position are calculated. A path consists of a sequence of points for a vehicle to follow, and guides vehicles to arrive at their desired position without collision.
\item Vehicle control. In order to realize reconfiguration of formation, vehicles plan their trajectories to pass through sequences of points, and calculate their control inputs with spatiotemporal constraints.
\end{enumerate}

The distribution calculation part is realized by the proposed iterative grouping algorithm, and will be discussed in Section.~\ref{conf}. The formation reconfiguration and vehicle control problems are solved by relative motion planning and control of vehicles, and will be explained in detail in Section.~\ref{forma}.

\begin{remark}
The proposed intersection management framework is applicable to unsignalized intersections either with or without flexible lane direction. The difference of applying it to intersections with and without flexible lane direction lies in the available position selection for vehicles in the distribution calculation part, by defining selectable lanes for vehicles to occupy when passing the conflict zone.
\end{remark}

%
\section{Conflict-free Distribution Calculation}
\label{conf}
%

In this section, we define conflict relationship between vehicles and propose a conflict-free distribution calculation method to spatially arrange vehicles to avoid collision when passing the conflict zone.

\subsection{Conflict definition}
\label{conflict}

A conflict represents that the motion of vehicles may lead to a collision. In this section, we focus on the conflict relationship between vehicles in the conflict zone of intersections. Previous research has defined conflicts in intersections with fixed lane direction~\citep{xu2018distributed}. In this paper, since lanes are set free for all kinds of turning, the conflict relationship is more complicated. As shown in Fig.~\ref{scenarioFig}, three types of conflict are defined:

\begin{enumerate}[(1)]
\item Crossing conflict. If the trajectories of two vehicles in the conflict zone cross, the two vehicles are in crossing conflict. A crossing conflict may exist between any two types of trajectories, including right-turning, straight-going, and left-turning, shown as the yellow rectangles in Fig.~\ref{scenarioFig}.
\item Diverging conflict. If the trajectories of two vehicles in the conflict zone start from the same lane, the two vehicles are in diverging conflict, shown as the yellow circles in Fig.~\ref{scenarioFig}.
\item Merging conflict. If the trajectories of two vehicles in the conflict zone end at the same lane, the two vehicles are in merging conflict, shown as the yellow triangles in Fig.~\ref{scenarioFig}.
\end{enumerate}

If two vehicles that are in conflict pass the conflict zone at the same time, a collision may happen. Based on the definition of movement of vehicles in (\ref{movement}), the conflict matrix is defined to check if vehicles with movement $i$ and $j$ are in conflict. The conflict matrix $\mathcal{C}$, whose element on the $i$-th row and the $j$-th column represents the conflict relationship between movement $i$ and $j$, is defined as:
\begin{eqnarray}
\mathcal{C}=[c_{i,j}]\in \mathbb{R}^{12n^\text{L}\times 12n^\text{L}},\ c_{i,j}=\begin{cases}
1, \ \text{if vehicle $i$ and $j$ are in conflict},\notag \\
0, \ \text{otherwise},
\end{cases},\ i,j=1,2,3,...,12n^\text{L}.
\end{eqnarray}

The number of rows and volumes of $\mathcal{C}$ is the same as the number of movements. After defining conflict relationship, the key is to schedule vehicles in conflict to pass the conflict zone at different time to prevent collision.

\subsection{Distribution calculation method}
\label{distribution}

Based on the defined conflict relationship, we need to calculate longitudinal and lateral distribution of vehicles to let them pass the conflict zone without collision. In order to regulate motion of vehicles and simplify planning, we make the following assumptions:

\begin{enumerate}[(1)]
\item Vehicles that are in conflict don't exist in the conflict zone simultaneously. 
\item Vehicles that are not in conflict and are planned to pass the conflict zone simultaneously arrive at the stop line at the same time.
\end{enumerate}

According to the above assumptions, vehicles are scheduled into layers. Vehicles that are in conflict are scheduled into different layers, and vehicles in the same layer have no conflict and pass the conflict zone simultaneously. The longitudinal gap between two layers is designed to let the former layer completely pass the conflict zone before the latter layer reach the stop line. Some methods have been proposed to calculate layers for vehicles in intersections with fixed lane direction,~\eg, depth-first spanning tree method, and maximum matching method. In this paper, we proposed the Iterative Grouping Algorithm (IGA) to calculate layers for vehicles in intersections with flexible lane direction, shown as Algorithm~\ref{iga}.

\begin{algorithm}[tb]
\label{iga}
\caption{The Iterative Grouping Algorithm (IGA)}  
\LinesNumbered  
\KwIn {Vehicle sequence with ID from 1 to $N$.
}
\KwOut{Lane distribution $\mathbb{L}$ and layer distribution $\mathbb{V}_k (k=1,2,3,...)$ for vehicles.
} 
\textbf{Initialization} Set $k\leftarrow0$, and set $\mathbb{L}$ as the current lane occupation of all the vehicles.\\
Calculate available layers for each vehicle.\\
\While{There is a vehicle not assigned to any layer.}{
Set $k\leftarrow k+1$ and $\mathbb{V}_k\leftarrow \emptyset$.\\
\For{$i\leftarrow 1$ \KwTo $N$}{
\If{Vehicle $i$ is not assigned to any layer {\bf and} layer $k$ is available for vehicle $i$.}{
Search all the possible lane occupation for vehicle $i$ and $\mathbb{V}_k$ and check conflict.\\
\If{Conflict-free lane occupation is found for vehicle $i$ and $\mathbb{V}_k$.}{
Set $\mathbb{V}_k\leftarrow \mathbb{V}_k\cup i$.\\
Update $\mathbb{L}$ as the one that has the minimum difference $e$ with the current $\mathbb{L}$.
}
}
}
}
\end{algorithm}  

The input of IGA is a sequence of vehicles with ID indexed from $1$ to $N$, and the output is conflict-free longitudinal (layer assignment) and lateral (lane occupation) distribution of vehicles. IGA firstly calculates available layers for all the vehicles, considering their maximum acceleration ability and the distance towards the conflict zone. The detail of calculating available layers for vehicles will be discussed later in Section~\ref{bilevel}. Then, IGA iteratively assigns vehicles to layers and check their conflict relationship. A new vehicle is assigned to a layer if it is not in conflict with any other vehicles already in the layer, considering their lane occupation. Finally, IGA outputs the conflict-free distribution of the $N$ vehicles. It is important to notice that IGA ignores vehicles' motion to avoid collision when calculating distribution, thus sometimes the calculated distribution is not feasible to conduct because some vehicles have to take extra steps to clear the way for other vehicles. The collision-free motion planning for multiple vehicles will be discussed in Section~\ref{forma}. If a distribution is found not feasible to conduct by multi-vehicle motion planning, the available layers for vehicles is adjusted and IGA is run again to calculate a new distribution. 

In IGA, the set $\mathbb{L}$ represents the lane occupation of vehicles, where the $i$-th element $\mathbb{L}_i$ represents the lane index that vehicle $i$ should occupy. From line 7 to line 11, IGA calculates lane occupation for vehicles. In line 10, the difference between two lane occupation results $\mathbb{L}$ and $\mathbb{L}'$ is calculated as follows:
\begin{eqnarray}
e=\sum_{i=1}^{N} (\mathbb{L}_i-\mathbb{L}'_i)^2
\end{eqnarray}

The lane occupation result that has the minimum difference with the current occupation is updated as the new lane occupation, in order to reduce unnecessary lane changes of vehicles.

\begin{theorem}
\label{igaproof}
The IGA has the following properties:
\begin{enumerate}[(1)]
\item Every vehicle will be assigned to a layer.
\item Vehicles that are assigned to the same layer are not in conflict with each other.
\item The algorithm will end in finite steps.
\end{enumerate}
\end{theorem}
\begin{proof}
The three properties are proved as follows:
\begin{enumerate}[(1)]
\item IGA iteratively assigns vehicles to layers and there are three situations:
\begin{enumerate}[a)]
\item The layer that a vehicle is trying to be assigned to is empty, so that no conflict will exist and the vehicle is successfully assigned to the layer.
\item The new vehicle is not in conflict with any other vehicles that are already in the layer, and the vehicle is successfully assigned to the layer.
\item Conflict exists between the new vehicle and other vehicles that are already in the layer, and this vehicle will not be assigned to this layer. Now we consider the worst case that the vehicle is in conflict with all the other vehicles in the intersection, so that it can't stay in the same layer with any other vehicle. Since the number of vehicles is finite in an intersection, the number of layers that the other vehicles are assigned to is also finite. After checking those finite number of layers, the new vehicle will be successfully assigned to the next available layer.
\end{enumerate}
This completes the proof of the first property.

\item This property is guaranteed by the assigning process. A vehicle will be assigned to a layer if:
\begin{enumerate}[a)]
\item The layer is empty and available for this vehicle. In this case, only one vehicle exists in the layer, and no conflict exists.
\item The vehicle is not in conflict with the vehicles that are already in the layer. In this case, no conflict exists between the new vehicles and the currently existing vehicles.
\end{enumerate}
Since vehicles are assigned to layers one by one, no conflict will exist between any two vehicles that are in the same layer. This completes the proof of the second property.

\item The algorithm will end if all the vehicles have been assigned to certain layers. As is provided in Proof (1), every vehicle will be assigned to a layer after finite steps, and the number of vehicles is finite. Thus, the total number of steps that are needed to assign all the vehicles to certain layers is also finite. This completes the proof of the third property.

\end{enumerate}

\end{proof}

\subsection{Example}
\label{example1}

An example is provided to show the result of IGA. A sequence of ten vehicles is given as input, and their traffic movement types are shown in Fig.~\ref{vehicleSequence}. Note that when calculating the distribution of vehicles for the example, we ignore the reachability of vehicles and all layers are available. 

\begin{figure}[b]
\begin{center}
    \includegraphics[width=0.5\linewidth]{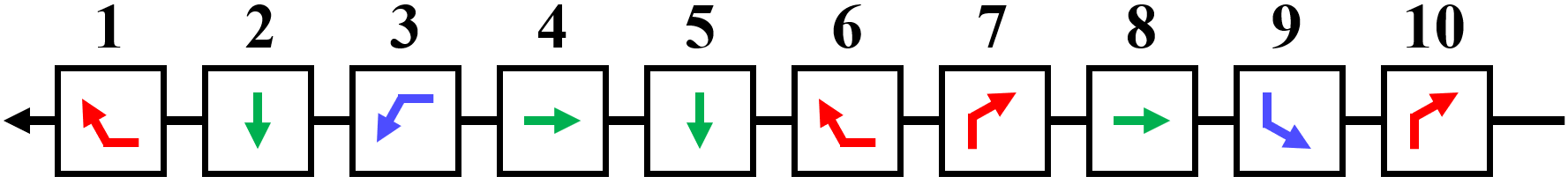}
    \caption{A sequence of ten incoming vehicles. The numbers are ID of vehicles and their traffic movements are shown in the black rectangles. Vehicles with the smaller ID are closer to the stop line of the intersection.}
    \label{vehicleSequence}
\end{center}
\end{figure}

In the intersection scenario with flexible lane direction, these vehicles are assigned into two layers, as shown in Fig.~\ref{variableresult}. In the first layer, seven vehicles including four right-turning vehicles, two straight-going vehicles, and one left-turning vehicle have no conflict and pass the conflict zone simultaneously. In the second layer, the other three vehicles, including two straight-going vehicles and one left-turning vehicle, have no conflict and pass the conflict zone simultaneously. In the intersection with fixed lane direction, where only one type of turning is allowed for each lane, four layers are needed to evacuate the ten vehicles, as shown in Fig.~\ref{fixedresult}. This indicates that cooperation at intersections with flexible lane direction can achieve fewer layers, shorter passing time, and higher efficiency.

%
\section{Multi-vehicle Formation Reconfiguration and Control}
\label{forma}
%

After being assigned to layers and lanes, vehicles need to plan trajectories to drive to their desired position. Collision may happen between vehicles that are driving on the same arm. Since the starting position and desired position of vehicles are known, we can utilize multi-vehicle formation reconfiguration method~\citep{cai2021formationc} to solve the conflict-free motion planning problem, by considering vehicles on the same arm as a group.

\begin{figure}
\begin{center}
    \subfigure[Distribution of vehicles]{
    \includegraphics[width=0.3\linewidth]{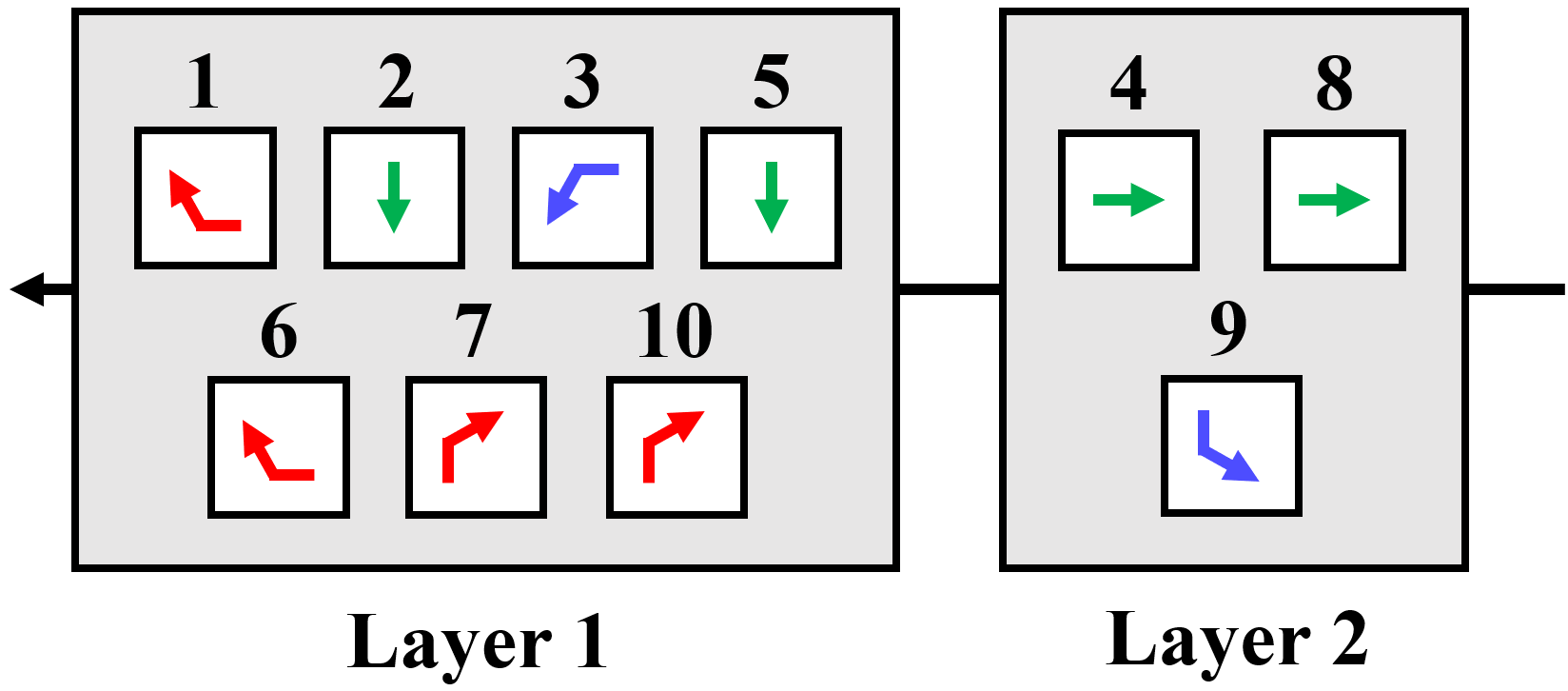}
    \label{distribution1}}
    \hspace{3mm}
    \subfigure[Layer 1]{
    \includegraphics[width=0.13\linewidth]{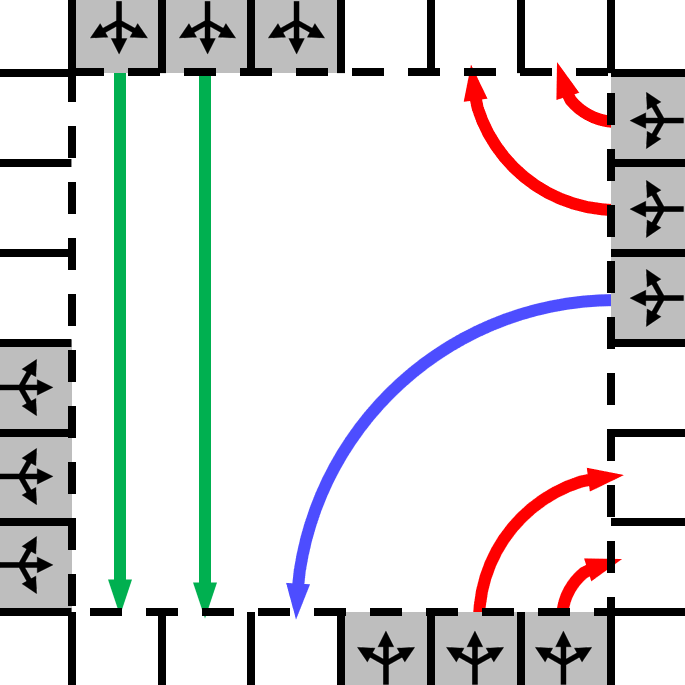}
    \label{layer11}}
    \hspace{3mm}
    \subfigure[Layer 2]{
    \includegraphics[width=0.13\linewidth]{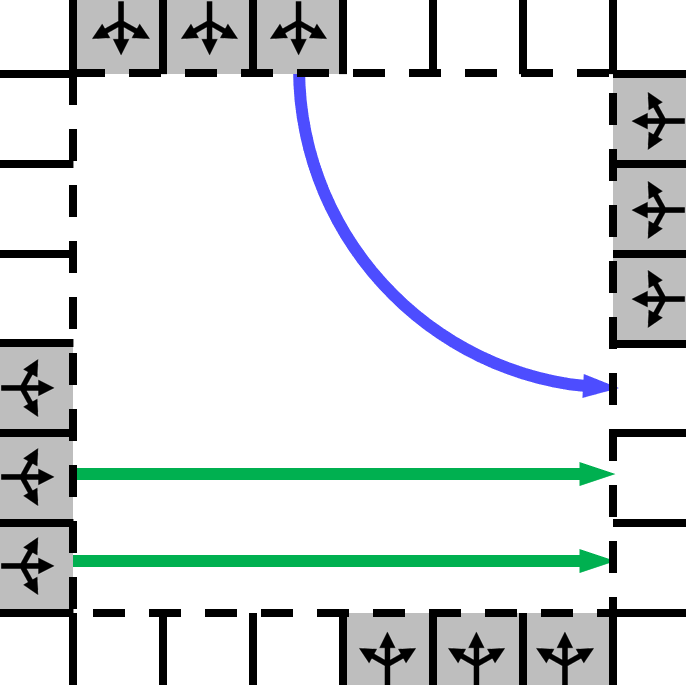}
    \label{layer12}}
    \caption{Distribution results of vehicles at intersections with flexible lane direction.}
    \label{variableresult}
\end{center}
\end{figure}

\begin{figure}
\begin{center}
    \subfigure[Distribution of vehicles]{
    \includegraphics[width=0.35\linewidth]{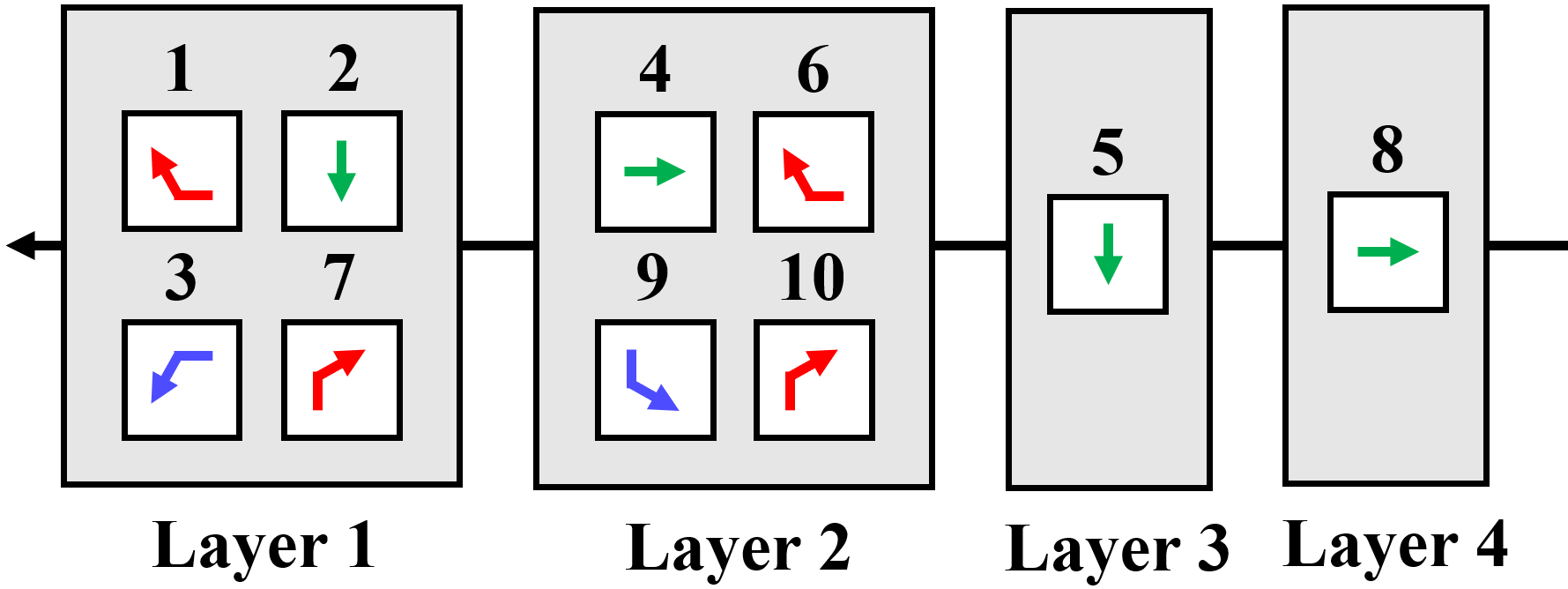}
    \label{distribution2}}
    \subfigure[Layer 1]{
    \includegraphics[width=0.13\linewidth]{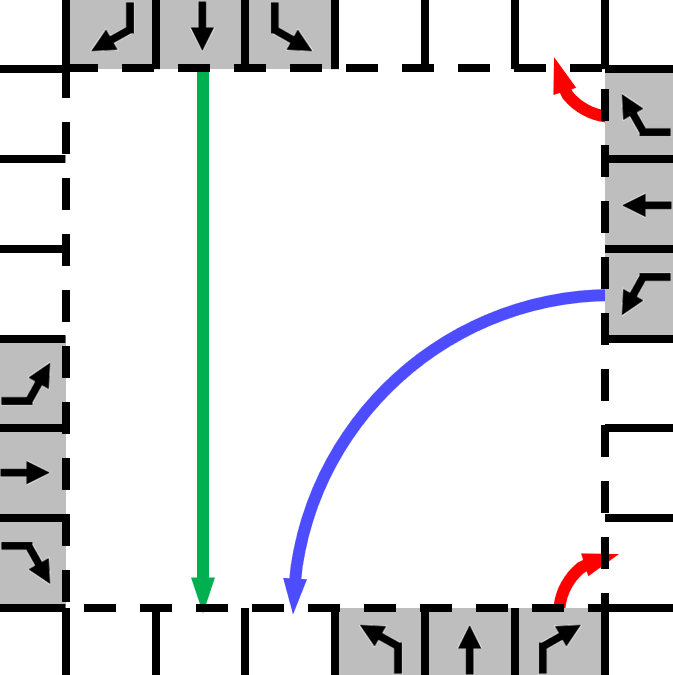}
    \label{layer21}}
    \subfigure[Layer 2]{
    \includegraphics[width=0.13\linewidth]{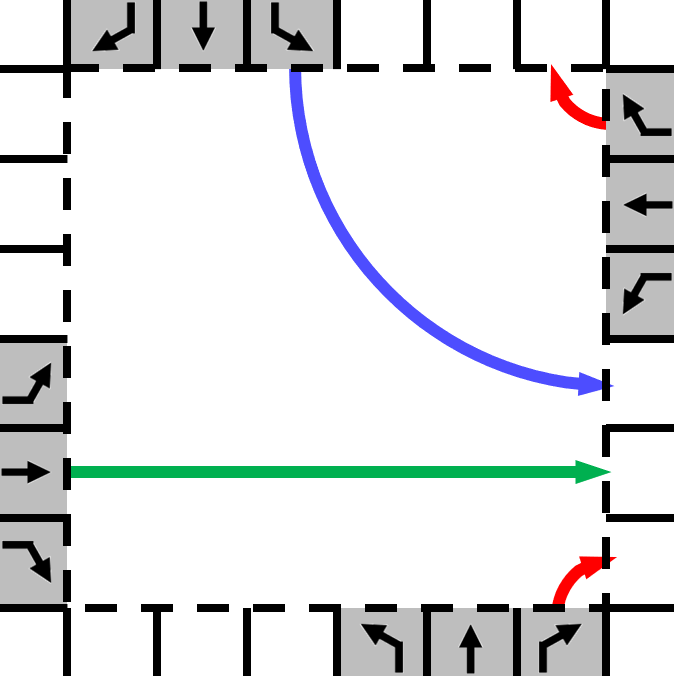}
    \label{layer22}}
        \subfigure[Layer 3]{
    \includegraphics[width=0.13\linewidth]{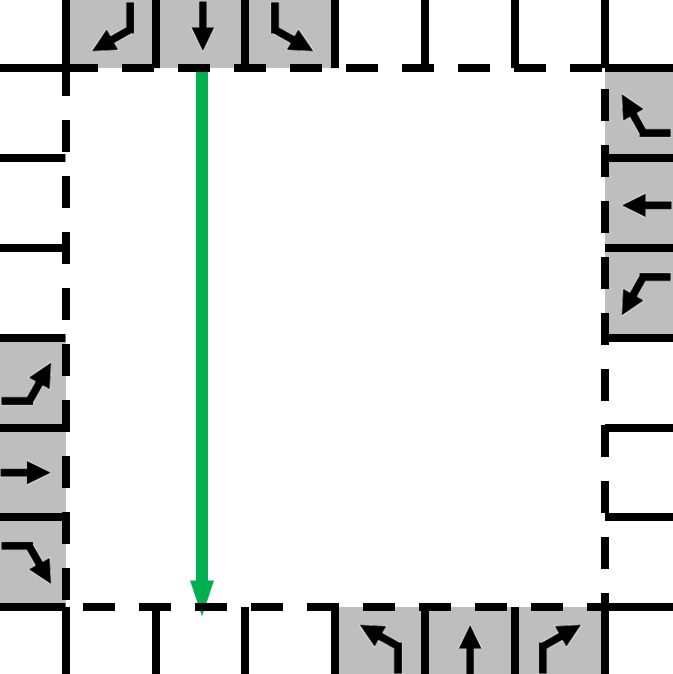}
    \label{layer23}}
        \subfigure[Layer 4]{
    \includegraphics[width=0.13\linewidth]{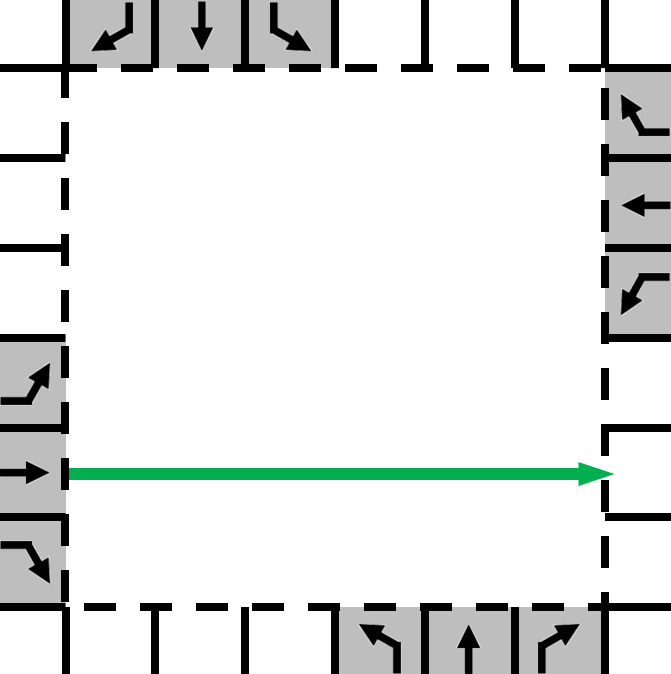}
    \label{layer24}}
    \caption{Distribution results of vehicles at intersections with fixed lane direction.}
    \label{fixedresult}
\end{center}
\end{figure}

\subsection{Bi-level formation control framework}
\label{bilevel}

The bi-level formation control framework is adopted to plan conflict-free trajectories for vehicles to catch up with the planned distribution, where relative motion planning is performed in the upper level to  avoid collision between vehicles, and trajectory planning with spatiotemporal constraints is conducted in the lower level to calculate control inputs for vehicles. The upper-level planning is performed in Relative Coordinate System (RCS), which is a dynamic coordinate system moving with vehicles. In RCS, position of vehicles is discretized by lanes laterally and by a constant gap $d_\text{F}$ longitudinally. Every relative point in RCS corresponds to a road point in Geodetic Coordinate System (GCS).

In order to prevent collision, time is also discretized by constant time intervals $T_\text{F}$. Vehicles occupy relative points in RCS with integer coordinates at fixed time points. During time period between two adjacent time points, vehicles can move from one relative point to another, or maintain their current relative position. If there are five available relative points for a vehicle at each time step (four points around and its current point), the motion map is four-connected. Given the relative position $(x^{\mathrm{r,v}}_{i,j},y^{\mathrm{r,v}}_{i,j})$ of vehicle $i$ at step $j$, the position $(x^{\mathrm{r,v}}_{i,j+1},y^{\mathrm{r,v}}_{i,j+1})$ at step $j+1$ is limited by:
\begin{eqnarray}
\label{4connected}
|x^{\mathrm{r,v}}_{i,j+1}-x^{\mathrm{r,v}}_{i,j}|+|y^{\mathrm{r,v}}_{i,j+1}-y^{\mathrm{r,v}}_{i,j}|\leq1.
\end{eqnarray}

If there are nine available relative points for a vehicle at each time step (eight points around and its current point), the motion map is eight-connected, and the motion of vehicles is regulated by:
\begin{eqnarray}
\label{8connected}
|x^{\mathrm{r,v}}_{i,j+1}-x^{\mathrm{r,v}}_{i,j}|\leq1,\\
|y^{\mathrm{r,v}}_{i,j+1}-y^{\mathrm{r,v}}_{i,j}|\leq1.
\end{eqnarray}

Denote the distance from the current position of a vehicle towards the stop line of the intersection as $d_i$ and desired speed of the formation as $v_\text{F}$, the maximum number of layers $N^\text{y}_i$ that vehicle $i$ can accelerate to take forward before reaching the stop line is:
\begin{eqnarray}
\label{forward}
N^\text{y}_i=\lfloor \frac{d_i}{v_\text{F}T_\text{F}+d_\text{F}}\rfloor,
\end{eqnarray}
where $\lfloor\cdot\rfloor$ is used to calculate the nearest integer that is not bigger than a number. Layers that are behind the most forward achievable layer for a vehicle are its available layers. The theoretical most forward available layers for all the vehicles is the most forward layer for the vehicle that is the closest to the stop line. Thus, the point whose relative coordinate $x^\mathrm{r}$ is 0 is set on this layer. Suppose that ID of the closest vehicle is 0, the relative coordinate $x^\mathrm{r,v}_i$ for vehicle $i$ is calculated as:
\begin{eqnarray}
x^\mathrm{r,v}_i=\left\langle\frac{d_i-d_0}{d_\text{F}}\right\rangle+N^\text{y}_0,
\end{eqnarray}
where $\left\langle\cdot\right\rangle$ is used to calculate the nearest integer of a number. Note that different vehicles may have the same $x^\mathrm{r,v}$ if they are close to each other longitudinally. In this case, they are forced to drive to identical relative points before the algorithm begins.

The relationship between relative states and real-world states are shown in Fig.~\ref{bilevelprocess}, where an example of movement during one time interval is provided. The red arrow in the upper part of Fig.~\ref{bilevelprocess} shows an oblique move (if the relative motion map is eight-connected), where the vehicle moves from relative coordinate (2,3) to (1,2) to join the other four vehicles and forms an interlaced five-vehicle formation. 

\begin{figure}
\begin{center}
    \includegraphics[width=0.5\linewidth]{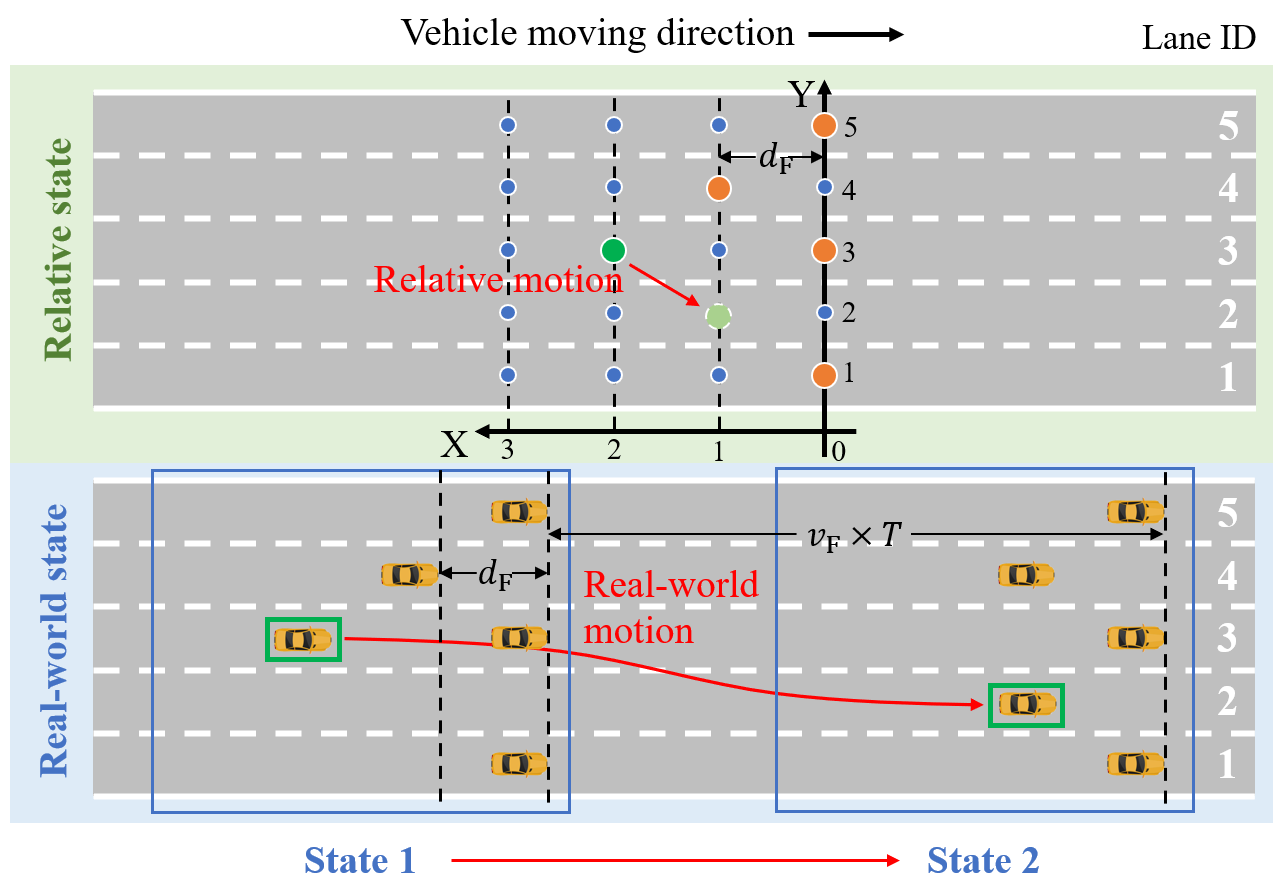}
    \caption{The bi-level formation control framework.}
    \label{bilevelprocess}
\end{center}
\end{figure}

Given the current and desired position of vehicles, we need to calculate collision-free trajectories for those vehicles to travel to their target position. This process is divided into two parts: the first part is to calculate conflict-free relative paths for vehicles in RCS, and the second part is to plan trajectories for vehicles to pass through those points one by one with spatiotemporal constraints. 

\subsection{Relative motion planning}
\label{relative}

The Conflict-based Searching (CBS) method is utilized to calculate collision-free relative paths for vehicles in RCS~\citep{cai2021formationc}. CBS method iteratively plan relative path for single vehicle, detect conflicts between vehicles, and re-plan paths for vehicles with constraints. A conflict exists between two vehicles and may lead to a collision. A constraint prevents a certain movement of a vehicle. The searching process is shown in the left part of Fig.~\ref{cTree}, where each node represents a complete cycle of single-vehicle path planning for all the vehicles in the formation. Conflicts are detected for each node and are added as constraints for the child nodes. Each node is evaluated by the total cost for vehicles to travel to their targets. The cost can be defined as the total travelling time, total travelled distance, etc. When finding a node without any conflict, a feasible conflict-free path solution is got. If all the other unsearched nodes have higher cost than a feasible node, this node is claimed to be the optimal solution with the minimum cost.

For each node in the searching tree, individual path planning is performed for every vehicle, considering the constraints. The path planning is conducted without considering other vehicles, and may cause conflicts, which is then added as the attributes of the node. The edge between nodes represents the generating order of different nodes. Two nodes are called as parent node and child node if there is a directed edge connecting them and starting from the parent node to the child node. All the constraints of a parent node should also be added to its child node, and one of the conflicts of the parent node will be added as a new constraint to its child node. That is to say, the number of child nodes of a node is equal to the number of the parent node's conflicts. Note that not every node is solid, because there may be conflicts and the corresponding paths calculated for the node will cause collision.

The tree starts from the root node $N_1$, which doesn't have constraints since it has no parent node. By applying single-vehicle path planning method to the initial and target position of vehicles, a path set, a conflict set, and a cost is got. The single-vehicle path planning algorithm utilized in this paper is the A* algorithm~\citep{hart1968formal}. Child nodes are generated to the root node and the tree iterates the process until an optimal solution is found. 

The CBS method is well known and widely used in the field of multi-robot path planning. The properties of CBS method including completeness, optimality, and time complexity have been proved and explained in numerous papers~\citep{sharon2013increasing,sharon2015conflict}, and we omit them in this paper.

\begin{figure}
\begin{center}
    \includegraphics[width=0.7\linewidth]{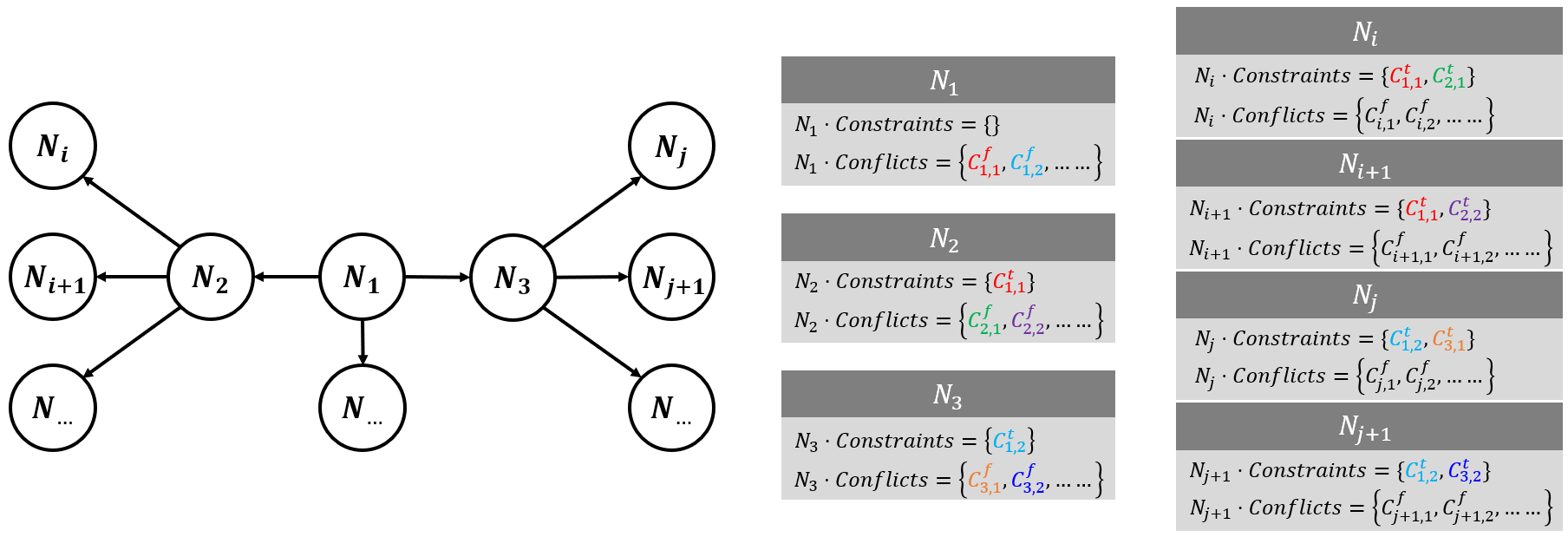}
    \caption{The process of conflict-based searching.}
    \label{cTree}
\end{center}
\end{figure}

\subsection{Multi-stage vehicle control}
\label{multistage}

Vehicles arrive at relative points at each $t_i$ in $\mathbb{T}=\{ t_i | t_i=iT, i=0,1,2,3...\}$ in RCS. B$\acute{\text{e}}$zier curves are generated for the vehicles to pass through the corresponding road points in GCS. The key is to calculate longitudinal and lateral control inputs for the vehicles to arrive at the desired position at desired time. The trajectory in GCS passes through a sequence of road points and the inaccuracy during the trajectory following process may accumulate by time. To prevent severe following error which may lead to collision, a multi-stage motion control framework is proposed and vehicles will replan their control inputs after desired time intervals.

The vehicle model used in this paper is the 2-DOF vehicle model. The center of the rear axle is chosen to represent the position of the vehicle. The yaw angle and the steer angle are represented as $\theta$ and $\delta$ respectively. $L$ represents the wheelbase of the vehicle. The state variable $\emph{\textbf{z}}^\mathrm{v}$ of the vehicle contains $x^\mathrm{v}$, $y^\mathrm{v}$, $v$ and $\theta$, where $v$ represents the velocity of the vehicle. The control input for the vehicle model contains the acceleration $a$ and the steer angle $\delta$, and the state is calculated as:
\begin{eqnarray}
&\dot{\emph{\textbf{z}}^\mathrm{v}}=
\begin{bmatrix} v\sin(\theta) \\ v\cos(\theta) \\ a \\ \frac{v}{L}\tan(\delta) \end{bmatrix},
&\emph{\textbf{z}}^\mathrm{v}=
\begin{bmatrix} x^\mathrm{v} \\ y^\mathrm{v} \\ v \\ \theta \end{bmatrix}.
\end{eqnarray}

For the lateral control, a linear feedback preview controller is designed. The vehicle looks ahead to the preview point and the closest point on the trajectory. The controller calculates the angle between the vehicle velocity and the desired velocity on the preview point, and the distance between the vehicle and the closest point as feedbacks.

The discretized longitudinal state equation is given as:
\begin{eqnarray}
\begin{cases}
&s(k+1)=s(k)+v(k)\text{d}t,\\
&v(k+1)=v(k)+a(k)\text{d}t,
\end{cases}
\end{eqnarray}
where $s(k)$, $v(k)$ and $a(k)$ represent the longitudinal position, speed and acceleration respectively. The time is discretized by sample interval $\text{d}t$. The distance that the vehicle will travel from the starting point to the final point is denoted as $S_t$. Since the vehicle moving in RCS should pass a sequence of relative points, $S_t$ consists of series of segments, whose length are denoted as $S_1$, $S_2$, ..., $S_{N^t}$, where $N^t$ is the number of segments. The time interval for the vehicle to cover each segment is $T_\text{F}$. Taking the control energy of the whole control process as the cost, the discretized longitudinal control can be described as:
\begin{alignat}{2}
\min\quad & \sum_{k=0}^{N^tN^k-1} a^2(k), &{}& \label{eqn - lon}\\
\mbox{s.t.}\quad
&s(0)=0,\notag \\
&s(k_i)=\sum_{n=1}^{i}S_n,\ k_i=iN^k,\ i=1,2,3,..., N^t,\notag \\
&v(0)=v(N^tN^k)=v_\mathrm{F},\notag\\
&v_{\text{min}}\leq v\leq v_{\text{max}},\notag\\
&a_{\text{min}}\leq a\leq a_{\text{max}},\notag
\end{alignat}
where $N^k=\frac{T_\text{F}}{\text{d}t}$ is the steps in one segment, $a(k)$ is the longitudinal control input (acceleration) of the vehicle at step~$k$, $v(k)$ is the longitudinal speed, $v_\mathrm{F}$ is the desired longitudinal speed of the formation, and $v_{\text{min}}$, $v_{\text{max}}$, $a_{\text{min}}$ and $a_{\text{max}}$ are the bounds of speed and acceleration.

The optimal control problem in (\ref{eqn - lon}) is solved and the sequence of control input is calculated to guide the vehicle to travel desired distance at desired time. Since the linear feedback preview controller is designed for lateral control, the inaccuracy of lateral motion may cause deviation for longitudinal motion. In order to resolve the accumulated inaccuracy, the optimal problem is reformed and solved every time the vehicle completes the following of one segment at time $t=iT_\text{F} (i=1,2,3,...,N^t)$. 

\begin{figure}
\begin{center}
    \subfigure[Process of formation reconfiguration]{
    \includegraphics[width=0.7\linewidth]{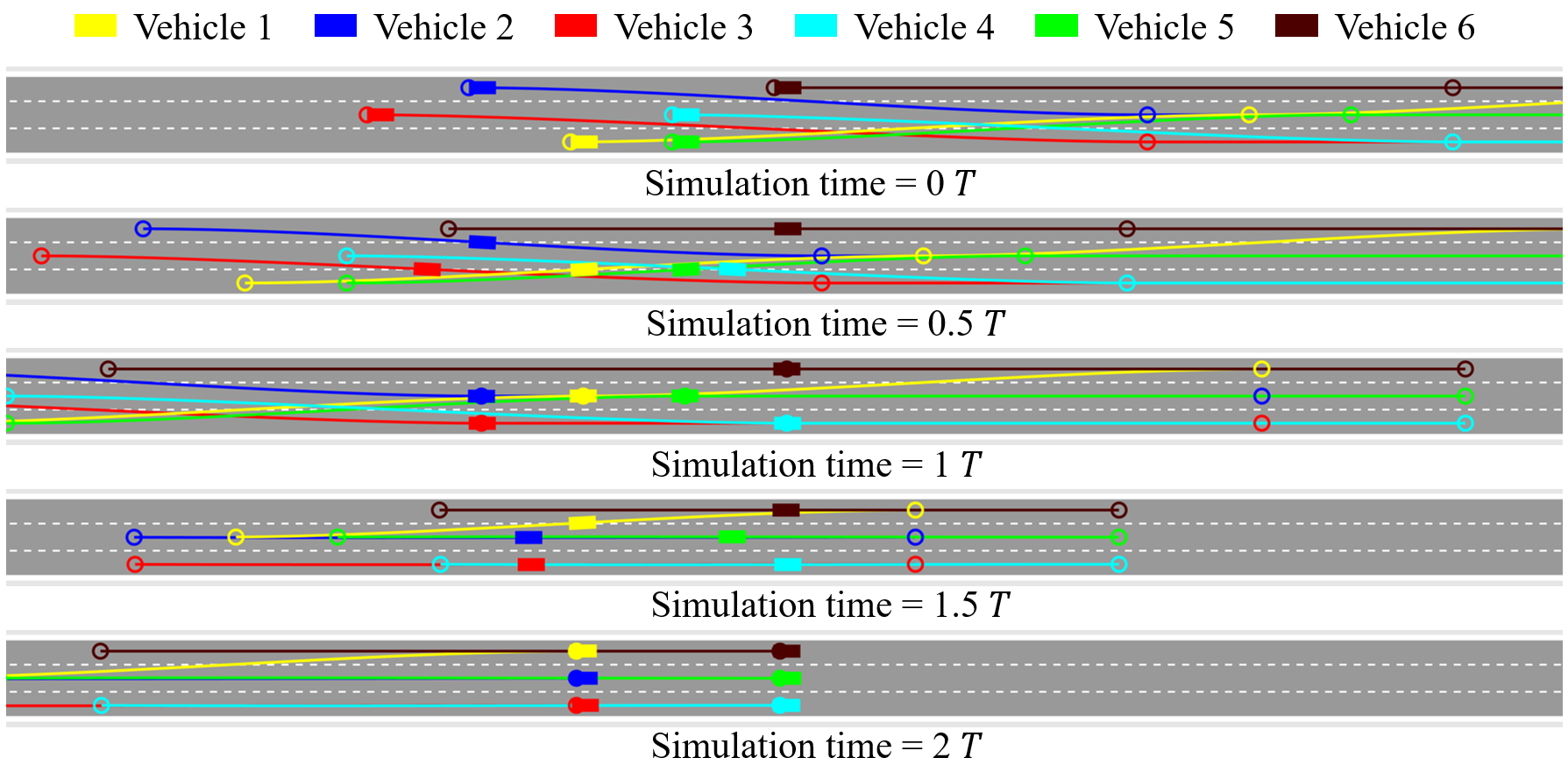}
    \label{case_movingprocess}}\\
    \subfigure[State 1]{
    \includegraphics[width=0.22\linewidth]{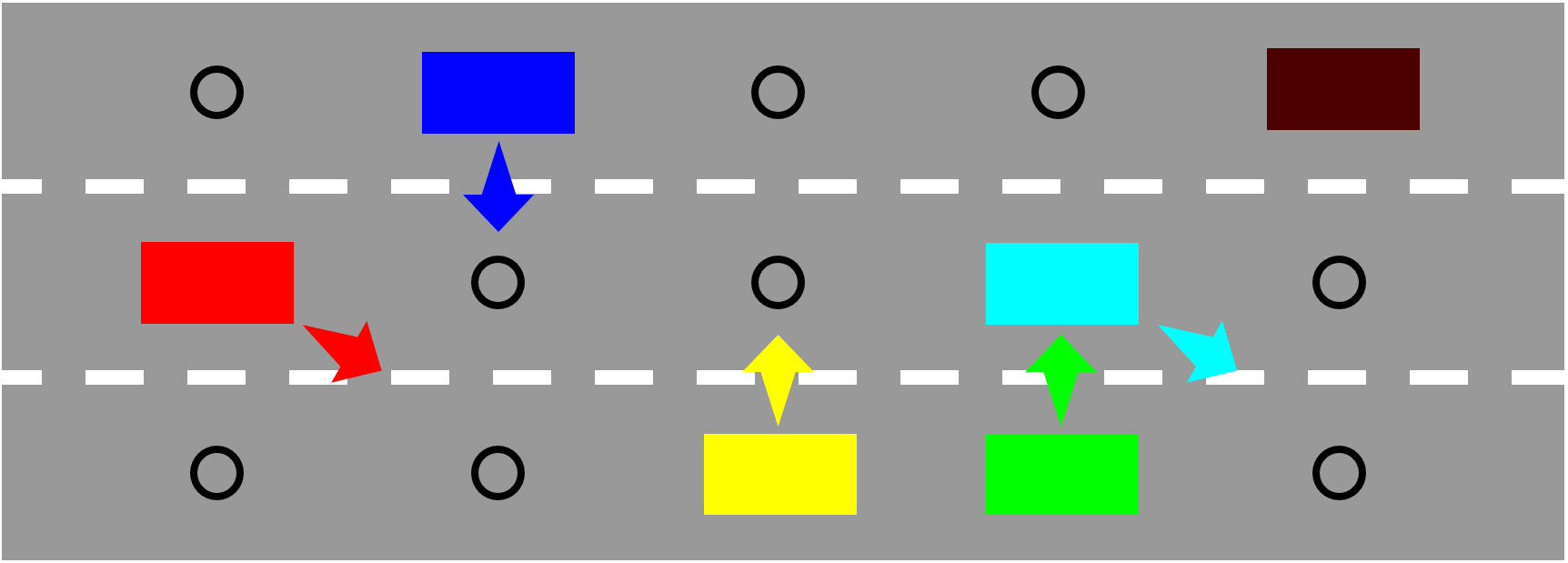}
    \label{case_state1}}
    \subfigure[State 2]{
    \includegraphics[width=0.22\linewidth]{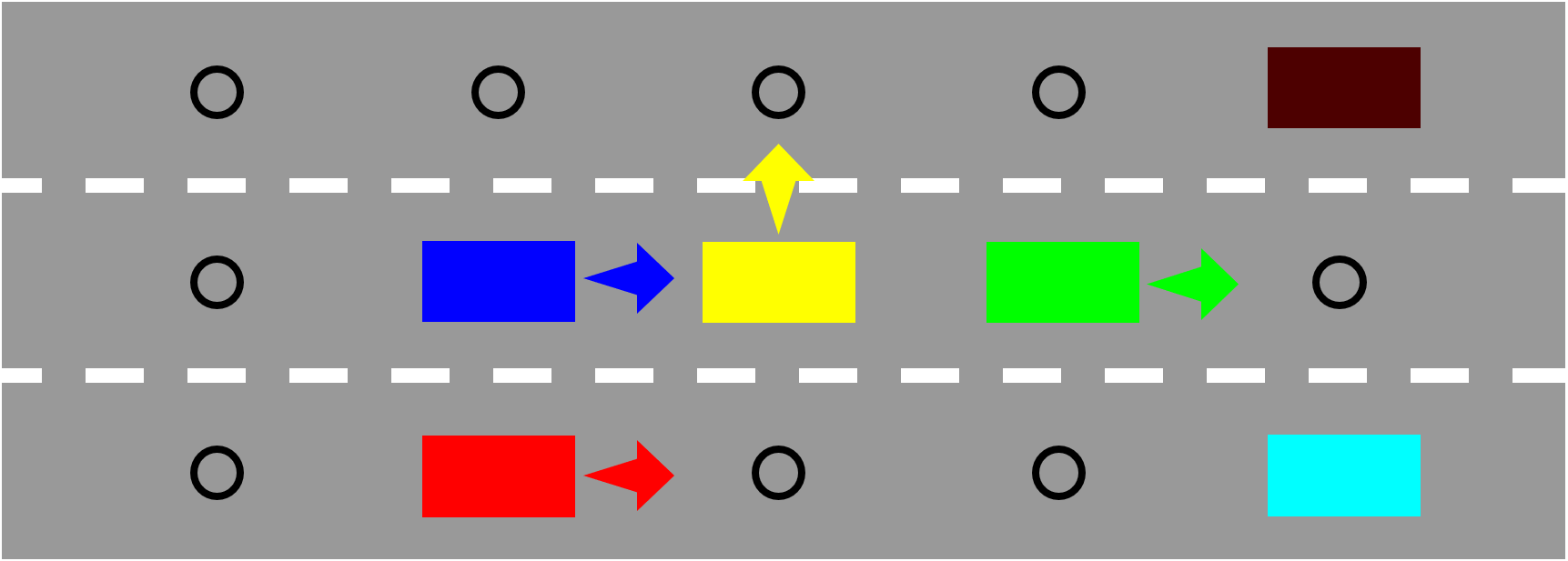}
    \label{case_state2}}
        \subfigure[State 3]{
    \includegraphics[width=0.22\linewidth]{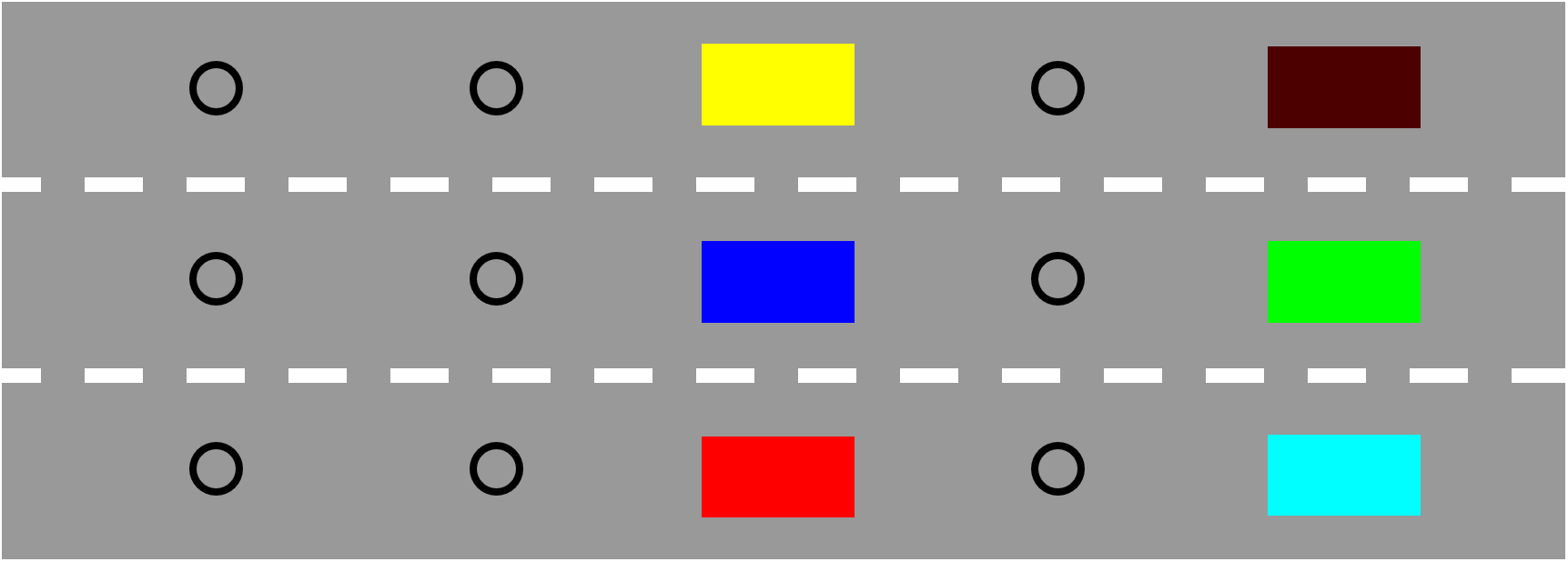}
    \label{case_state3}}
    \caption{Formation reconfiguration process of six vehicles.}
    \label{case}
\end{center}
\end{figure}

\subsection{Example}
\label{example2}

An example is provided in Fig.~\ref{case} to show the process of formation reconfiguration. Six vehicles start from messy points and need to form a parallel three-lane formation. The initial and expected position of vehicles are shown in Fig.~\ref{case_state1} and Fig.~\ref{case_state2}. The motion map for this example is eight-connected, and the reconfiguration process takes two cycles. The rectangles represent vehicles and the lines with the same color represent trajectories of certain vehicles. Fig.~\ref{case_movingprocess} shows the state and trajectories of vehicles in GCS, and the circles represent road points for vehicles to pass through. Fig.~\ref{case_state1}, Fig.~\ref{case_state2}, and Fig.~\ref{case_state3} show the states and movements of vehicles in RCS, and the black circles represent discretized relative points.

%
\section{Simulations}
\label{simu}
%

In this part, the simulation is conducted and the results are analyzed. The settings of the simulation are presented, and the benchmark methods are introduced. Then, numerical results and snapshots are presented to evaluate the proposed method.

\subsection{Simulation settings}
\label{settings}

The simulation is implemented with MATLAB 2017b and SUMO 0.32.0~\citep{lopez2018microscopic} on a personal computer with CPU Intel CORE i7-8700@3.2GHz. The average travelling time of vehicles is analyzed to evaluate the performance. The parameters chosen for this simulation are presented in Table~\ref{para}. It is important to notice that the following distance of two layers $d_\text{c}$ when passing the conflict zone is set twice as the safe following distance in formation. Hence, there is a transitional relative point between two layers for vehicles to temporally stay when moving in a formation. The optimization problem in (\ref{eqn - lon}) is solved by the GPOPS solver~\citep{patterson2014gpops}.

\begin{table}[htbp]
\centering
\caption{Simulation parameters}
\label{para}
\begin{tabular}{lll}
\toprule
Safe one-lane following gap in formation                        &   $d_\text{F}$          & $17.5\,\mathrm{m}$ \\ 
Safe one-lane following gap in conflict zone                        &   $d_\text{c}$          & $35\,\mathrm{m}$ \\ 
Formation switching cycle                          &   $T_\text{F}$              & $4\,\mathrm{s}$ \\ 
Desired speed in formation                          &   $v_\text{F}$              & $10\,\mathrm{m/s}$\\ 
Minimum speed of vehicle                          &   $v_{\text{min}}$              & $0\,\mathrm{m/s}$ \\ 
Maximum speed of vehicle                          &   $v_{\text{max}}$              & $20\,\mathrm{m/s}$\\ 
Minimum acceleration of vehicle              &   $a_{\text{min}}$              & $-10\,\mathrm{m/s^2}$ \\ 
Maximum acceleration of vehicle                 &   $a_{\text{max}}$          & $5\,\mathrm{m/s^2}$ \\ 
\bottomrule  
\end{tabular}
\end{table}

The scenario of simulation is a three-lane intersection with four incoming arms, as shown in Fig.~\ref{scenarioFig}. The length of each arm is $400\,\mathrm{m}$. We conduct simulation in different input traffic volumes, from $1000\,\mathrm{vehicle/hour}$ to $5000\,\mathrm{vehicle/hour}$. Vehicles are generated randomly on different lanes. At each input traffic volume, we conduct simulation with four different input proportion of vehicles with different turning expectation, where the proportion of right-turning vehicles, straight-going vehicles, and left-turning vehicles is $(0.33,0.33,0.34)$, $(0.25,0.25,0.5)$, $(0.25,0.5,0.25)$, and $(0.5,0.25,0.25)$ respectively. At each input volume and turning proportion, we compare the results of three scenarios: unsignalized intersection with flexible lane direction, unsignalized intersection with fixed lane direction, and signalized intersection with fixed lane direction. The signal timing cycle is given in Fig.~\ref{signal} and time for each phase is $20\,\mathrm{s}$. In order to clear vehicles in conflict zone after each phase, a three-second clearing yellow-light phase is conducted between each of the four phases. Similar signal timing cycle design can be found in~\cite{xu2018cooperative}.  Random inputs are generated for each cycle of simulation, and the results are the average of ten cycles. In total, $200$ vehicles are generated for each case, where the first $100$ vehicles are used for simulation warm up, and we only count the travelling time of the last $100$ vehicles for result comparison.

\begin{figure}
\begin{center}
    \subfigure[Phase 1]{
    \includegraphics[width=0.18\linewidth]{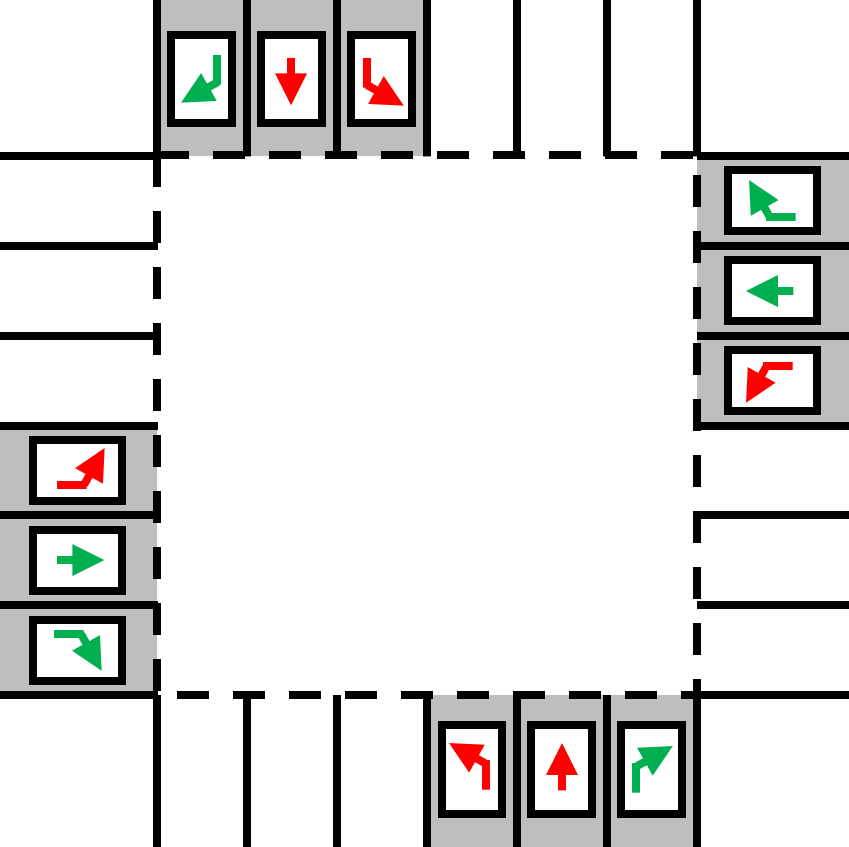}
    \label{movingprocess}}
    \subfigure[Phase 2]{
    \includegraphics[width=0.18\linewidth]{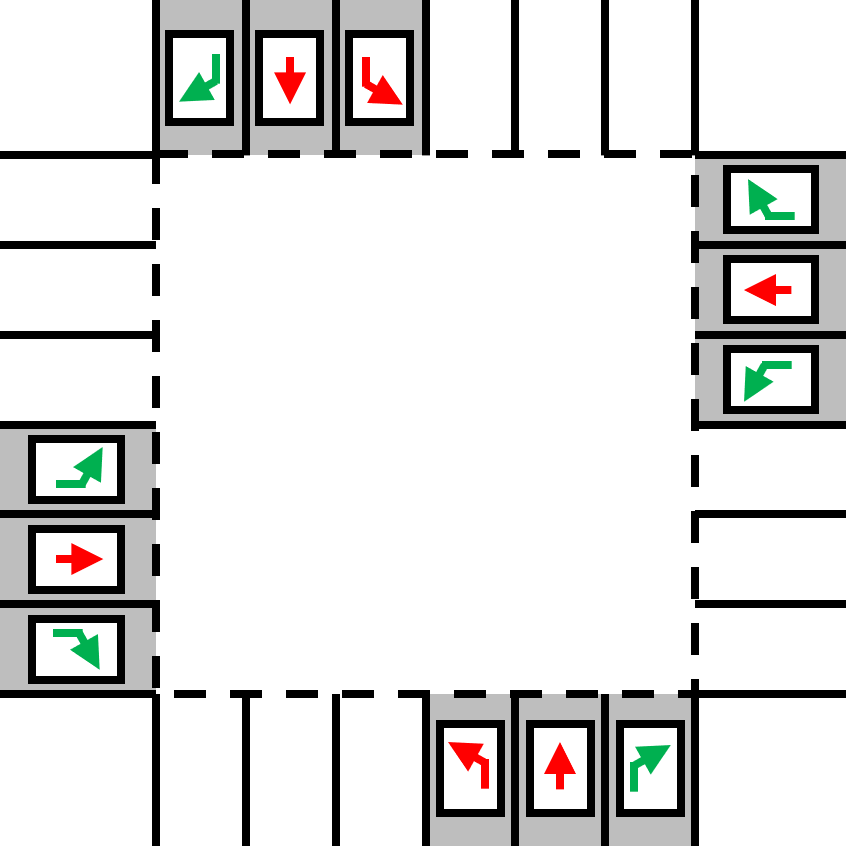}
    \label{state1}}
    \subfigure[Phase 3]{
    \includegraphics[width=0.18\linewidth]{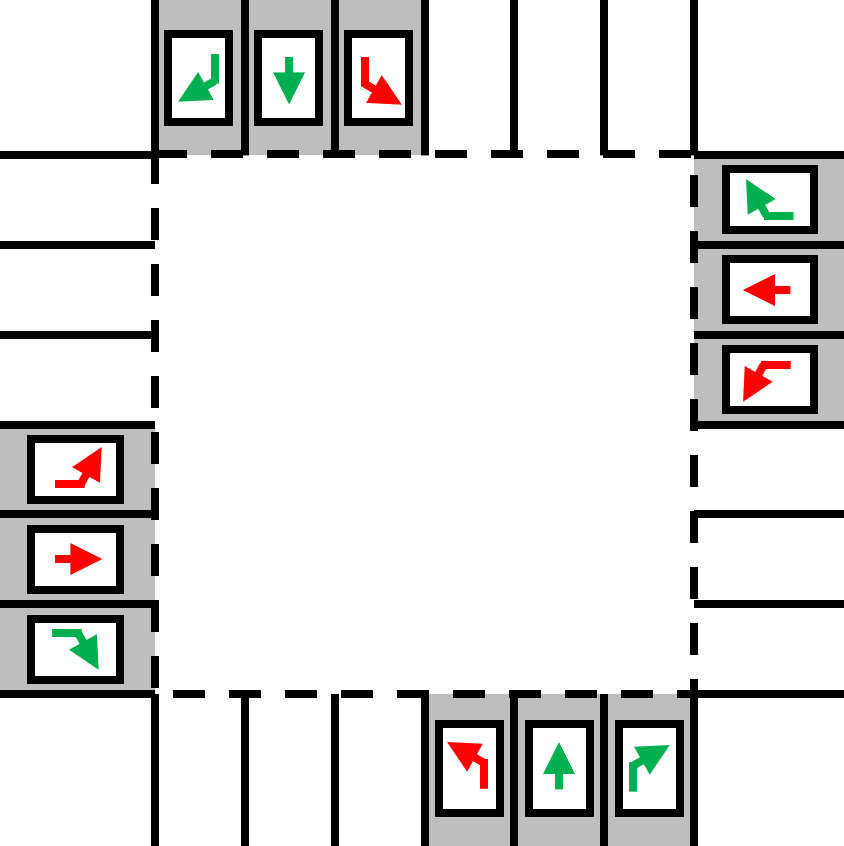}
    \label{state2}}
    \subfigure[Phase 4]{
    \includegraphics[width=0.18\linewidth]{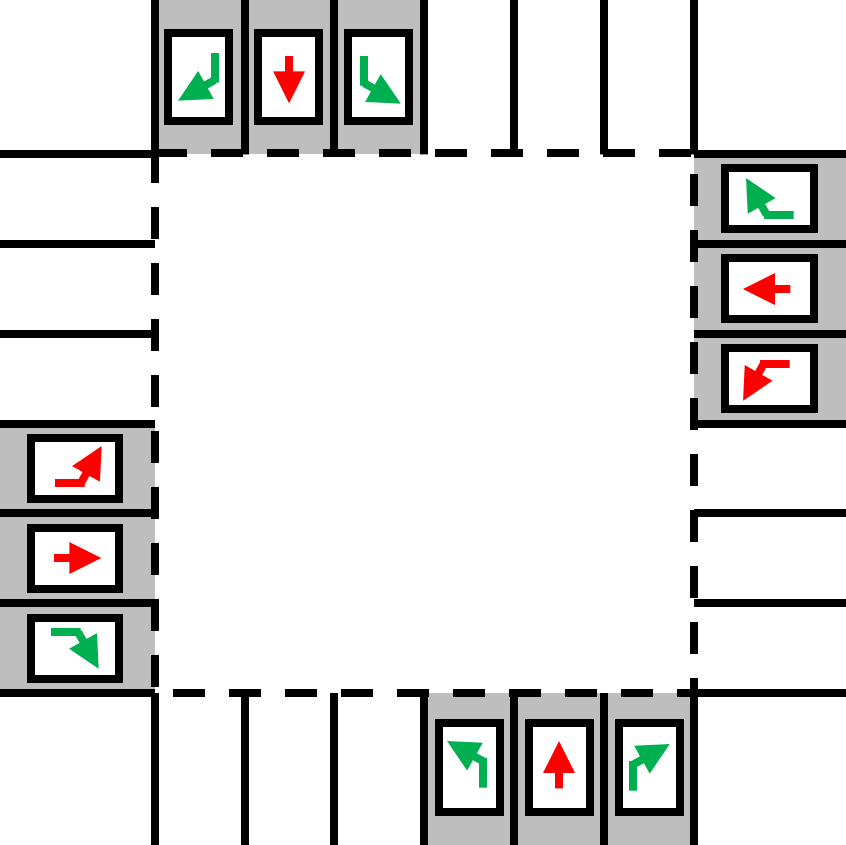}
    \label{state3}}
    \caption{Four phases of signal timing cycle.}
    \label{signal}
\end{center}
\end{figure}

\subsection{Results}
\label{results}

The results of average travelling time that vehicles take to travel on the incoming arms are shown in Fig.~\ref{att}. The four figures are the results under four turning proportion of inputs. The standard deviation under different random inputs is also marked in the figures. From the figures we can see that our method at flexible-lane-direction scenarios achieves that best performance, under all the input volumes and turning proportion. The reason is that our method allows more vehicles to pass the conflict zone simultaneously than the other two methods. At low traffic volumes (from 1000 to $3000\,\mathrm{vehicle/hour}$) the unsignalized cooperation methods outperform the signalized method, because vehicles don't have to stop when arriving at the conflict zone and can maintain a higher speed in the whole journey. As the input volumes becomes higher, the performance of the unsignalized cooperation with fixed lane direction become worse than the signalized method, because vehicles have to slow down or even brake to achieve their desired layers to pass the intersection. It is important to notice that the performance of the methods differs at different input proportion, but the proposed method always has the shortest travelling time results.

\begin{figure}[b]
\begin{center}
    \subfigure[Proportion $(0.33,0.33,0.34)$]{
    \includegraphics[width=0.225\linewidth]{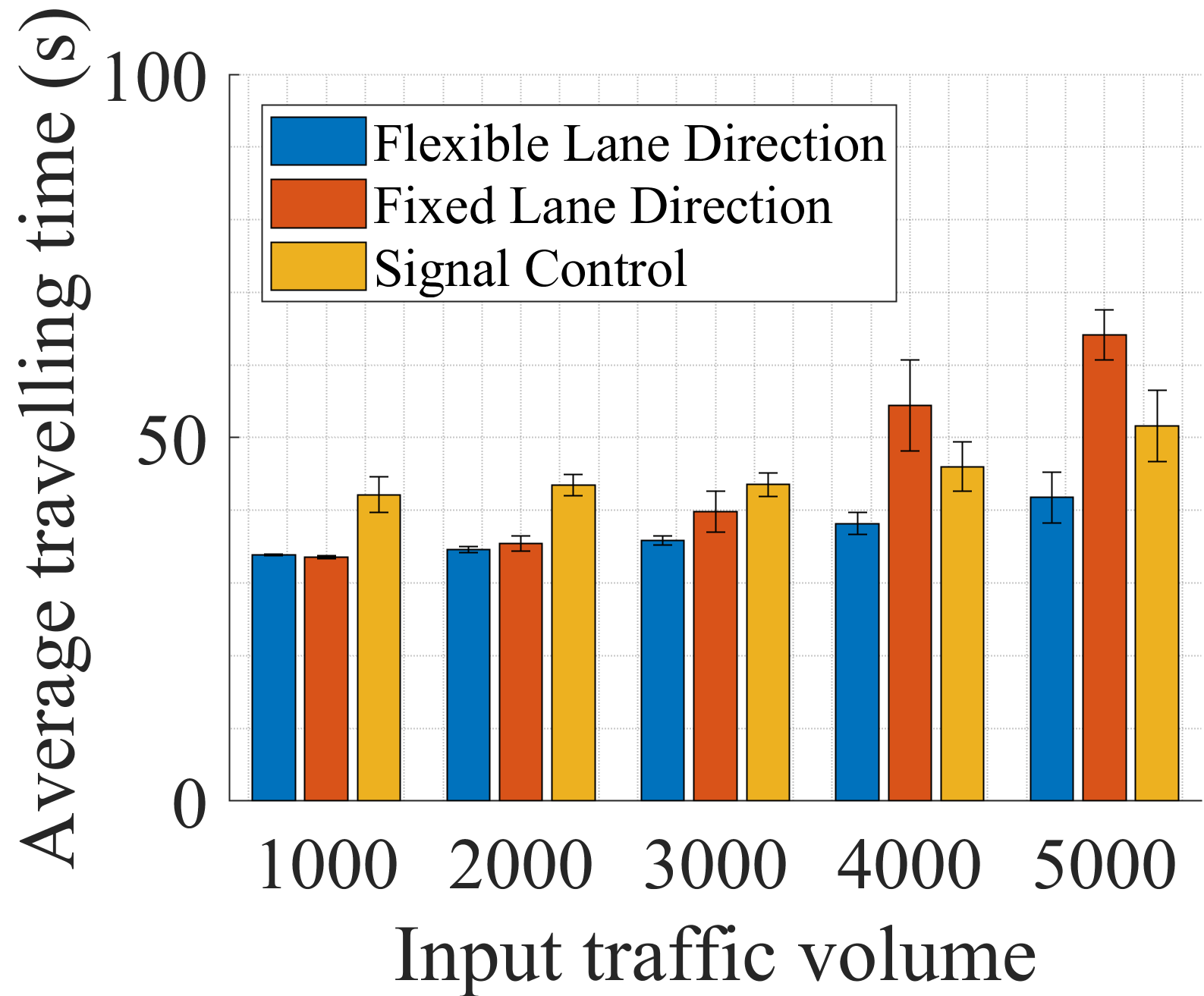}
    \label{ratio1}}
    \subfigure[Proportion $(0.25,0.25,0.5)$]{
    \includegraphics[width=0.225\linewidth]{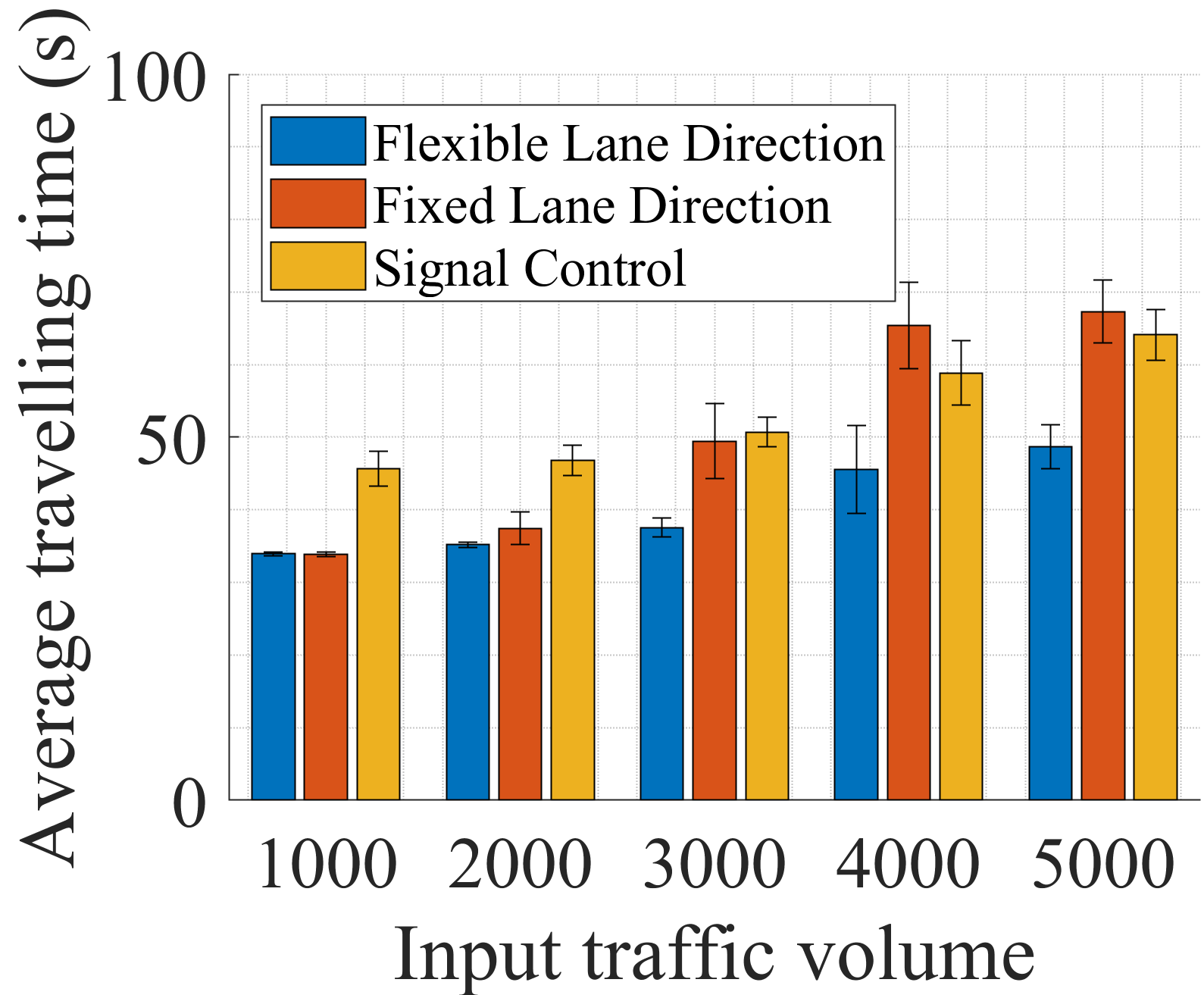}
    \label{ratio2}}
    \subfigure[Proportion $(0.25,0.5,0.25)$]{
    \includegraphics[width=0.225\linewidth]{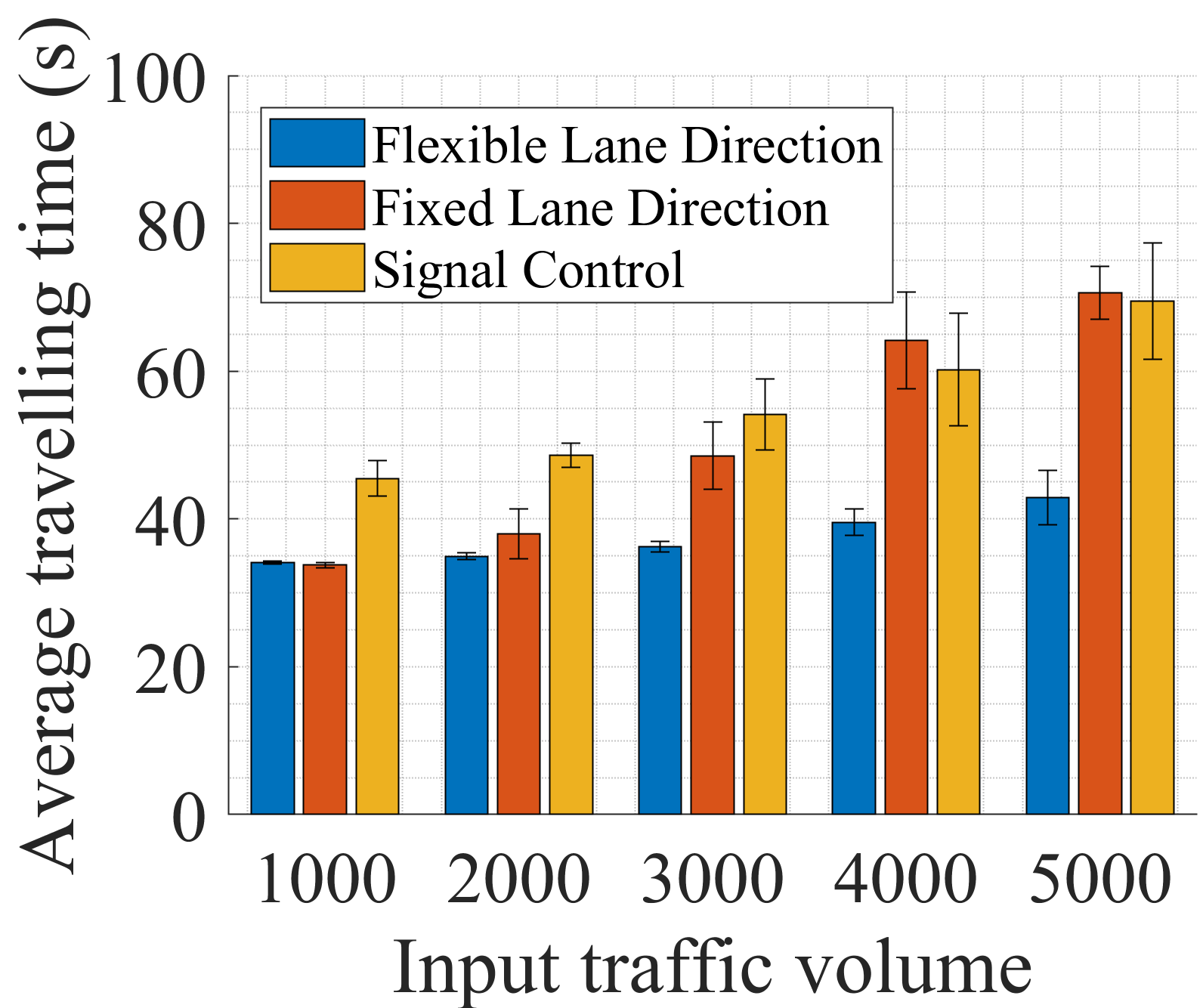}
    \label{ratio3}}
    \subfigure[Proportion $(0.5,0.25,0.25)$]{
    \includegraphics[width=0.225\linewidth]{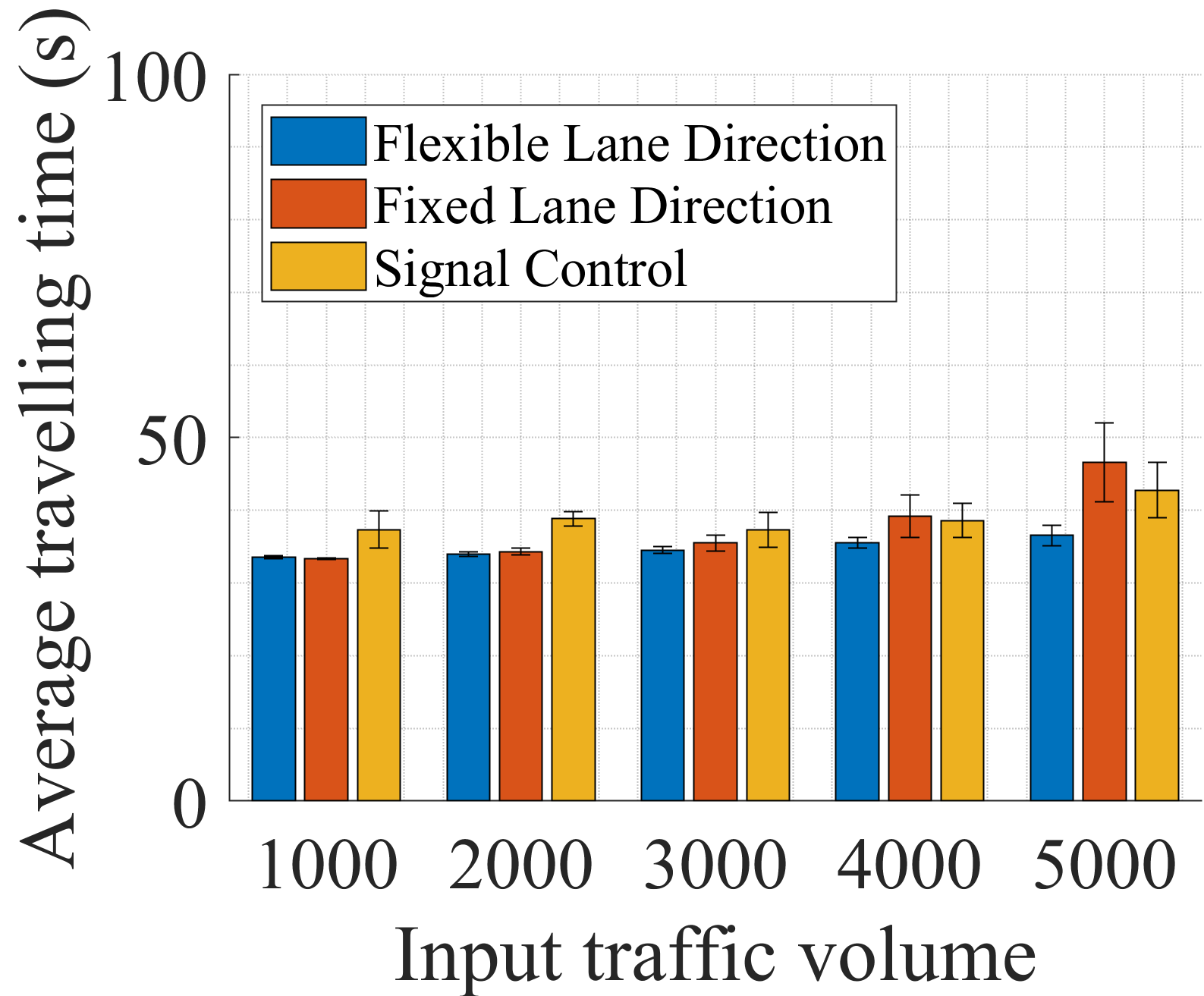}
    \label{ratio4}}
    \caption{Results of average travelling time of vehicles.}
    \label{att}
\end{center}
\end{figure}

\begin{figure}
\begin{center}
    \includegraphics[width=0.95\linewidth]{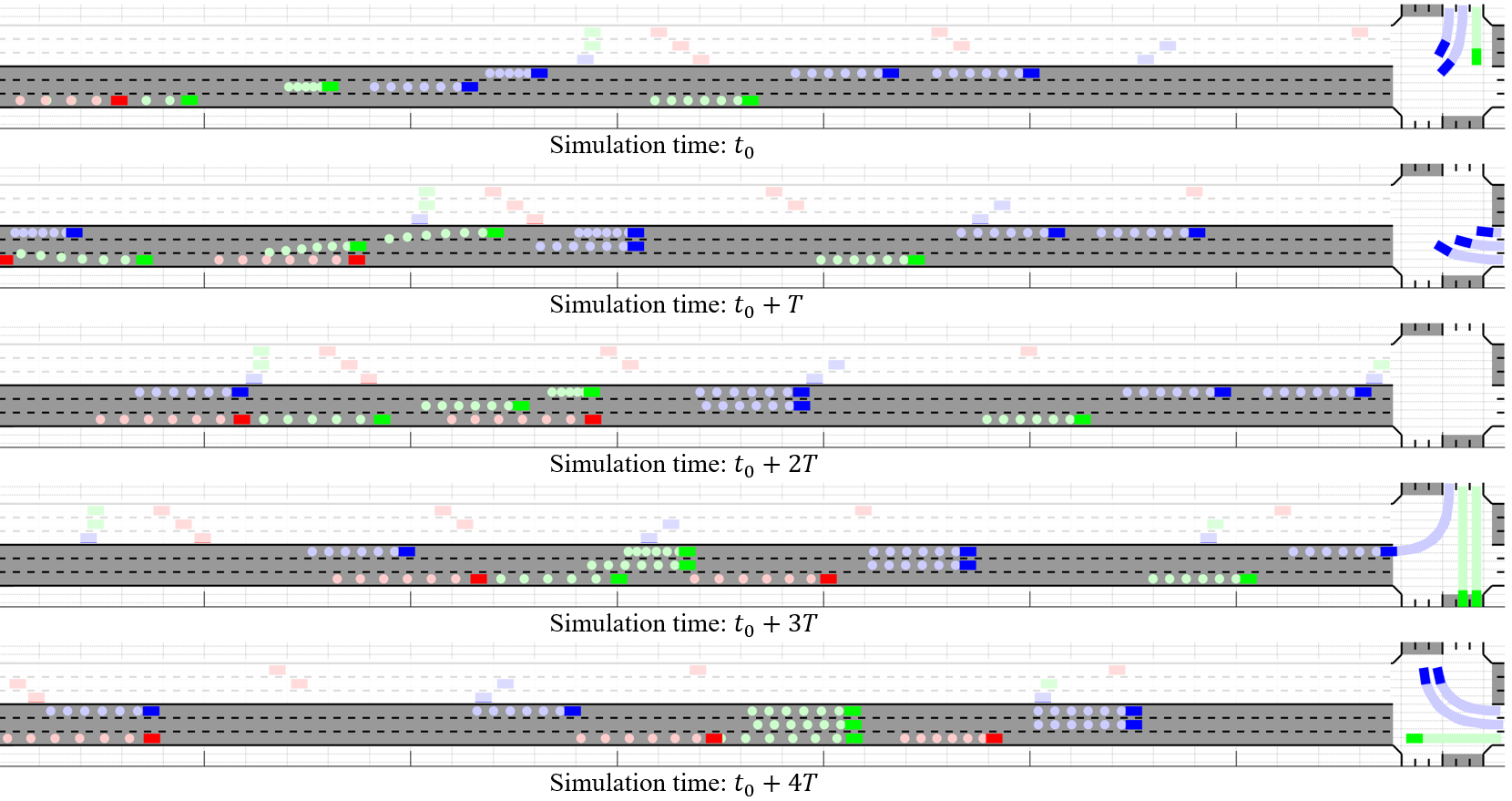}
    \caption{Snapshots of the west incoming arm at four continuous cycles of the simulation. The rectangles represent vehicles and their color shows their turning expectation, where red for right-turning, green for straight-going, and blue for left-turning. The full version video is available online at the address: {\color{blue}https://github.com/cmc623/flexible-lane-direction-intersection}}
    \label{segment}
\end{center}
\end{figure}

\begin{figure}
\begin{center}
    \subfigure[]{
    \includegraphics[width=0.225\linewidth]{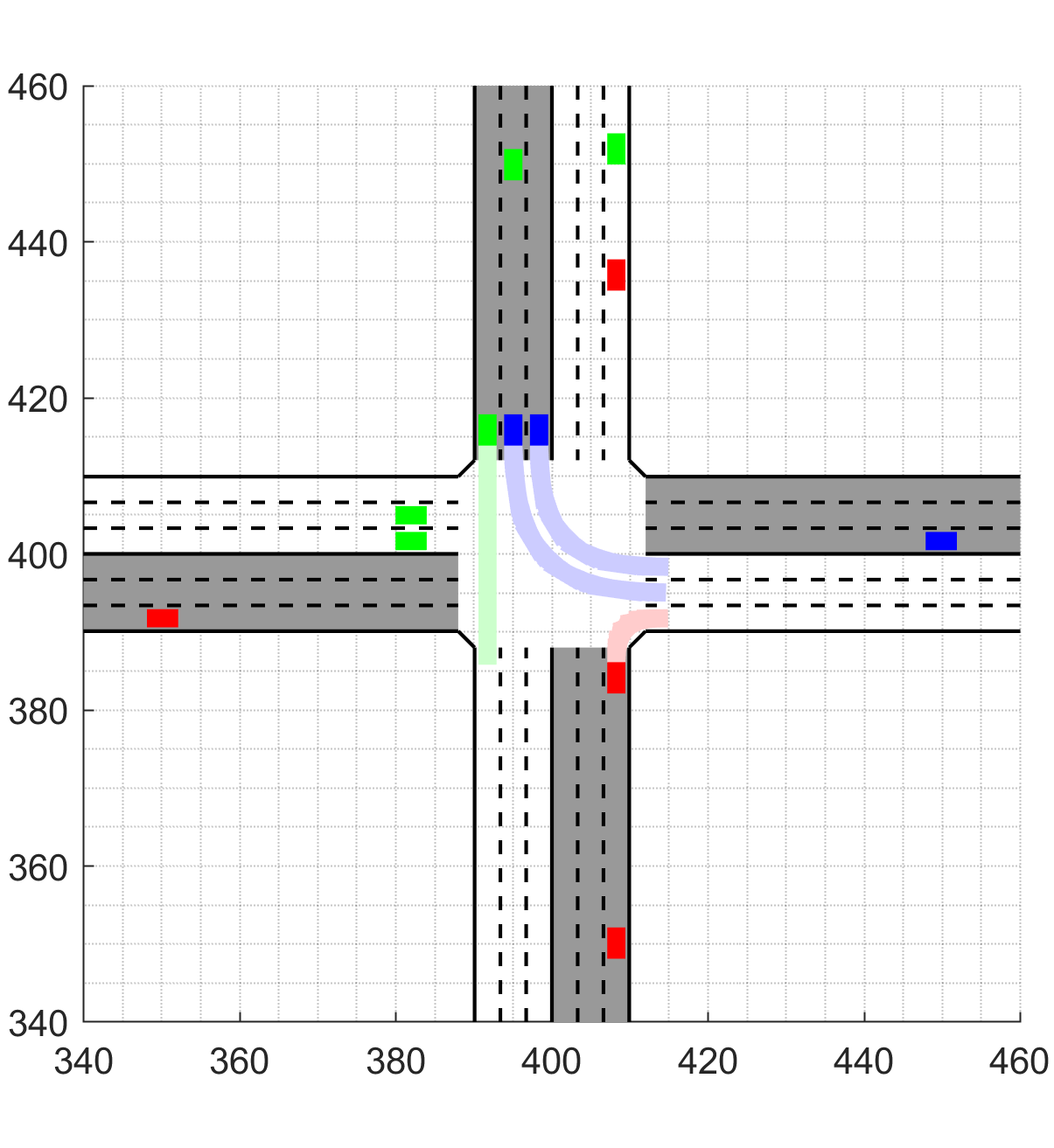}
    \label{1_460}}
    \subfigure[]{
    \includegraphics[width=0.225\linewidth]{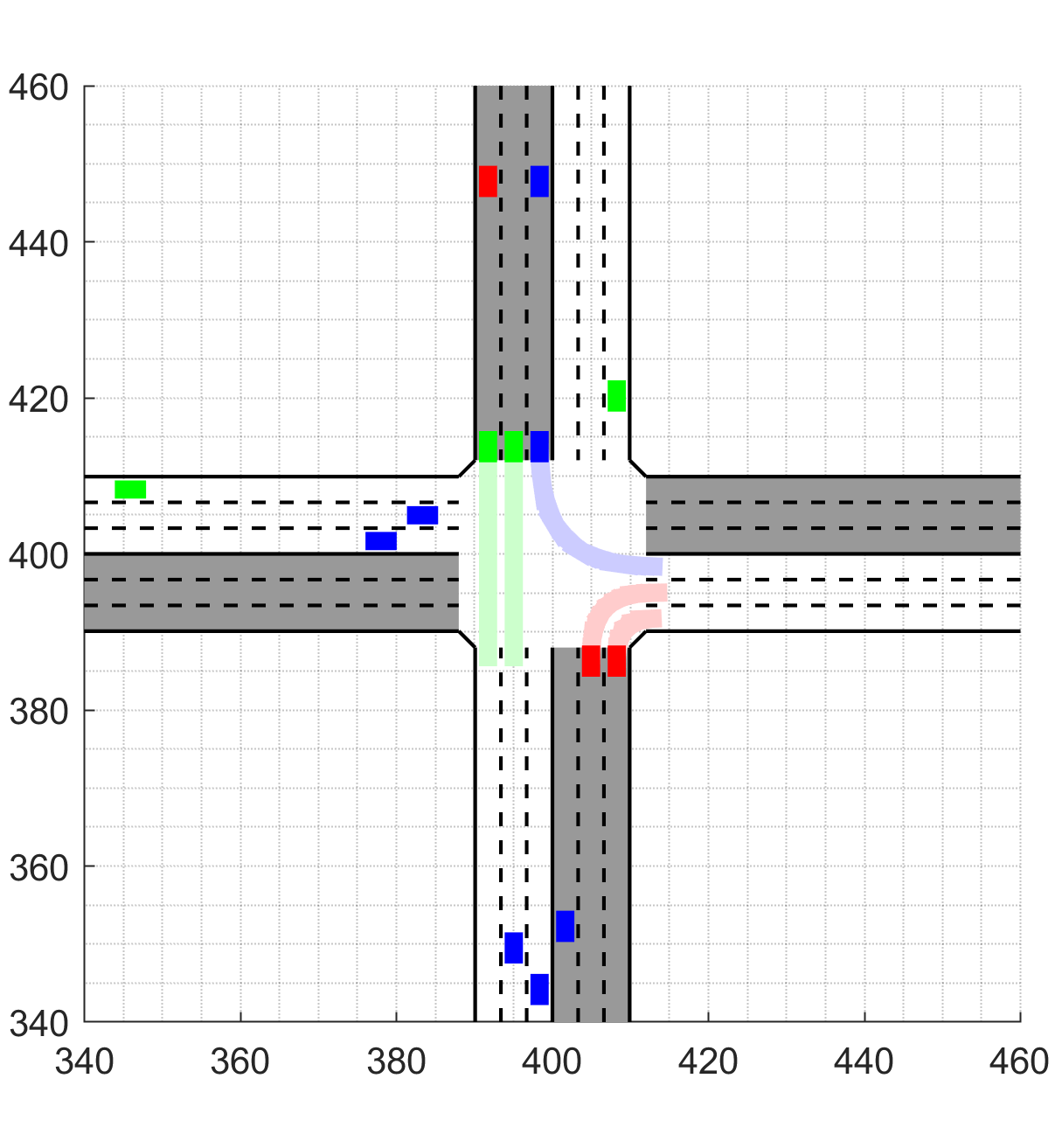}
    \label{1_970}}
    \subfigure[]{
    \includegraphics[width=0.225\linewidth]{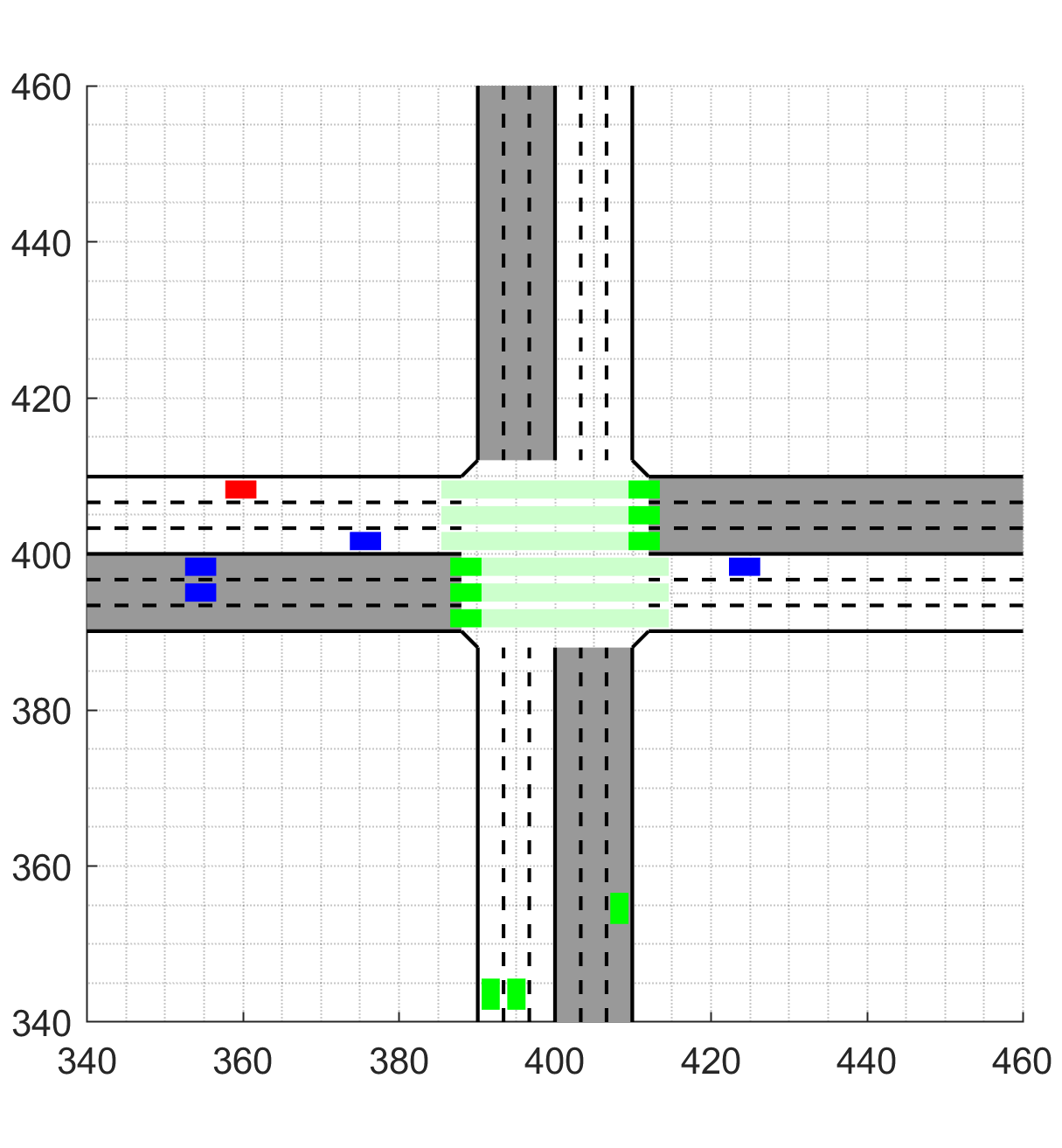}
    \label{1_1040}}
    \subfigure[]{
    \includegraphics[width=0.225\linewidth]{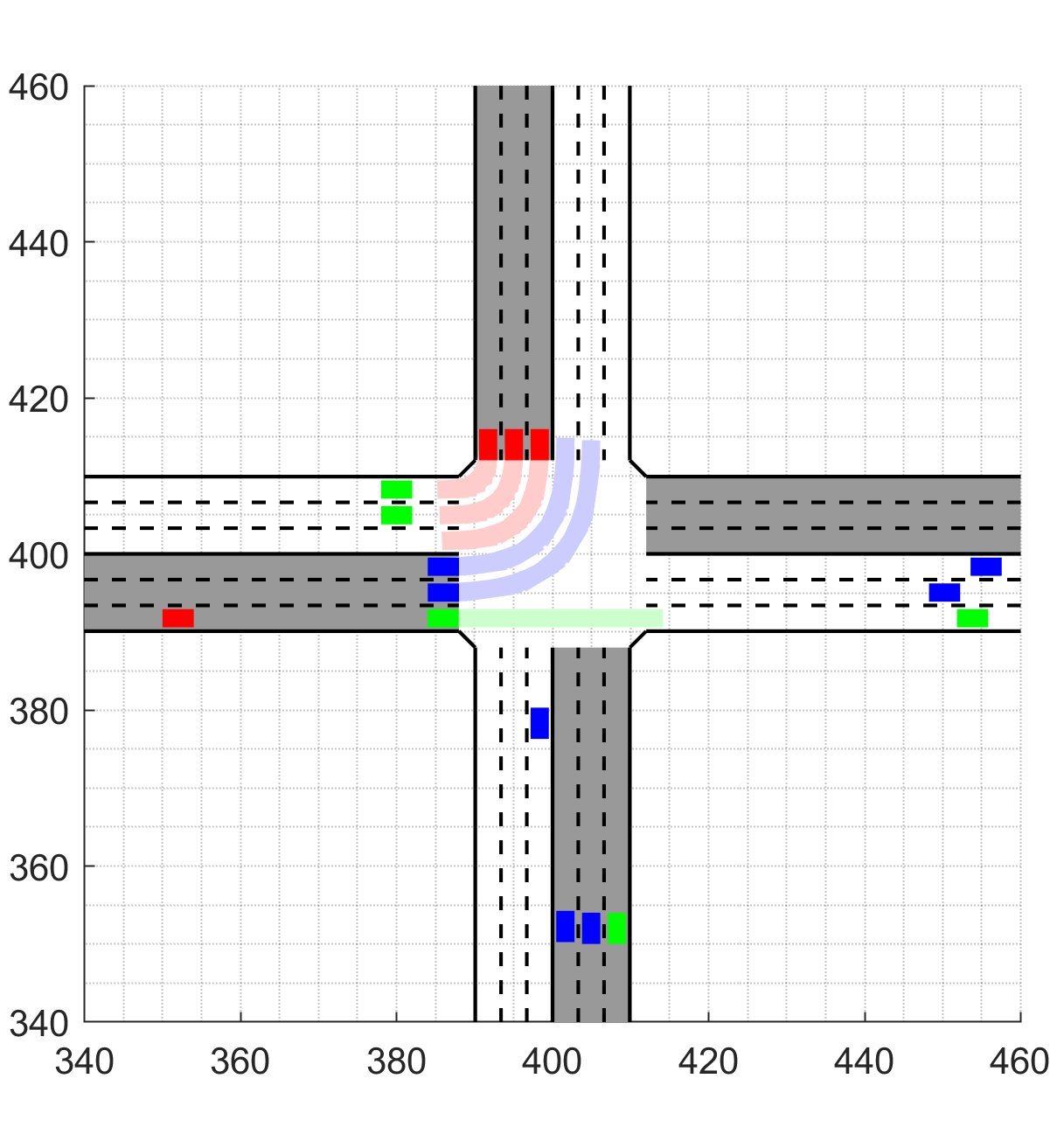}
    \label{1_1850}}
    \subfigure[]{
    \includegraphics[width=0.225\linewidth]{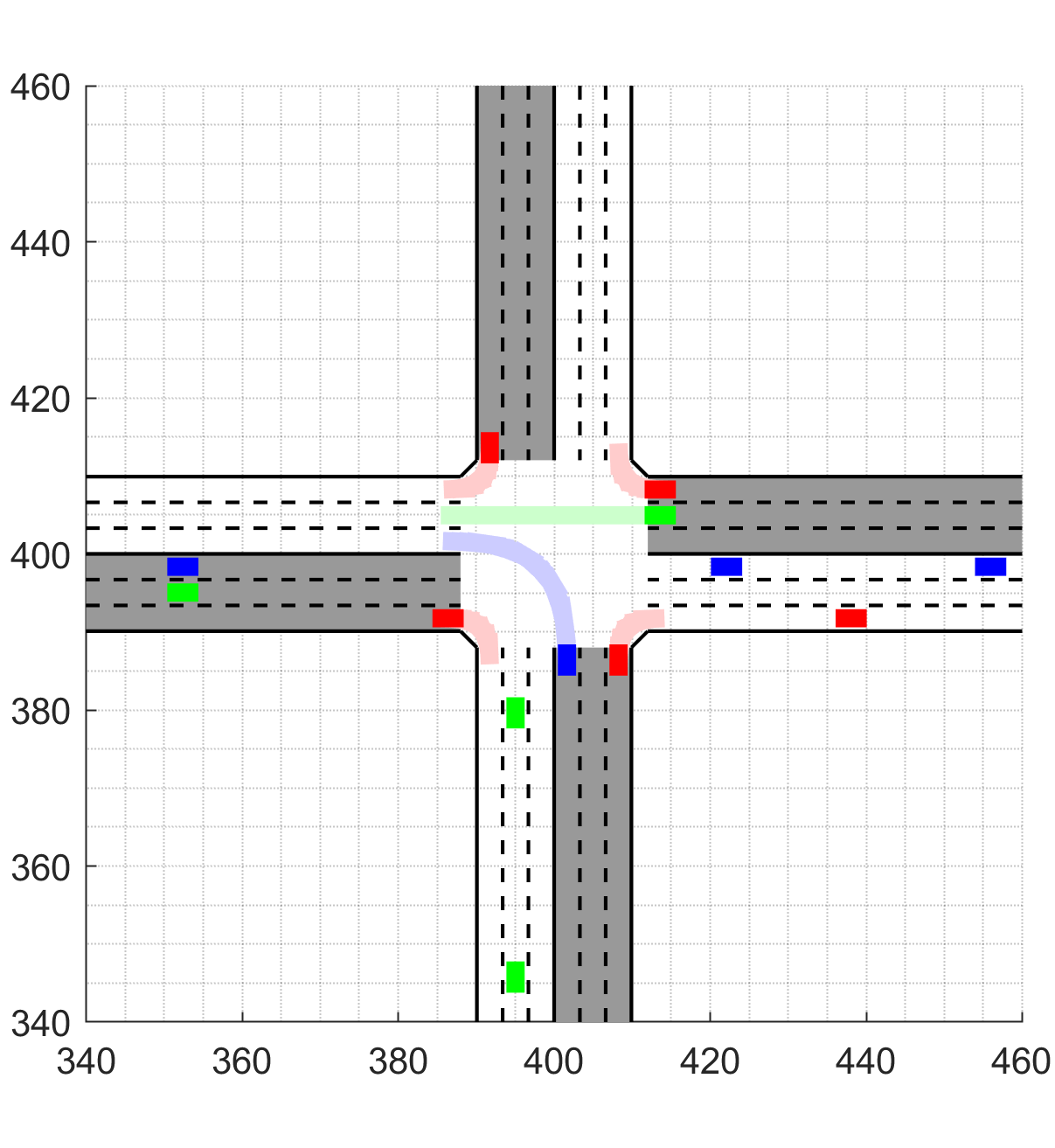}
    \label{2_530}}
    \subfigure[]{
    \includegraphics[width=0.225\linewidth]{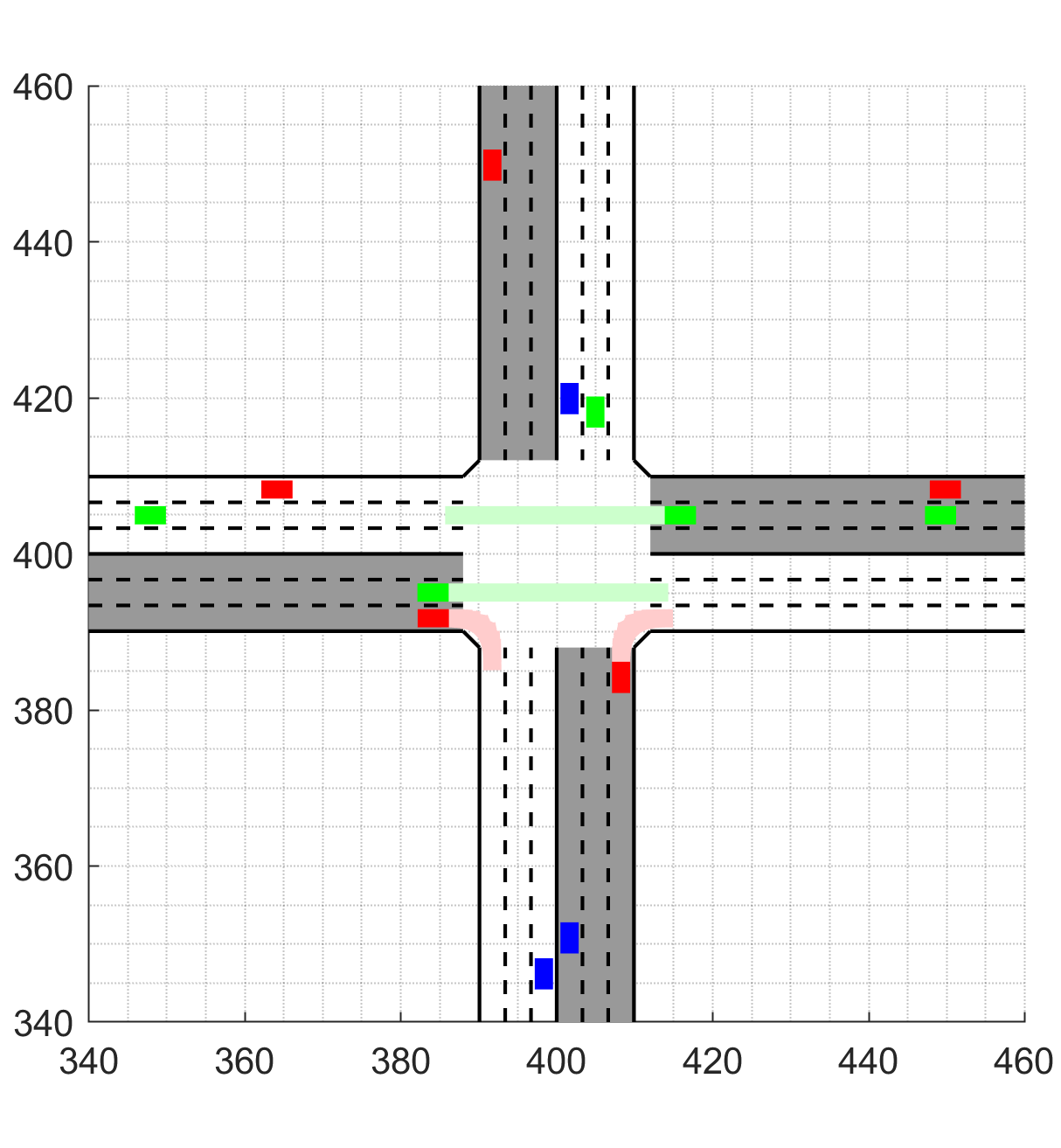}
    \label{2_1440}}
    \subfigure[]{
    \includegraphics[width=0.225\linewidth]{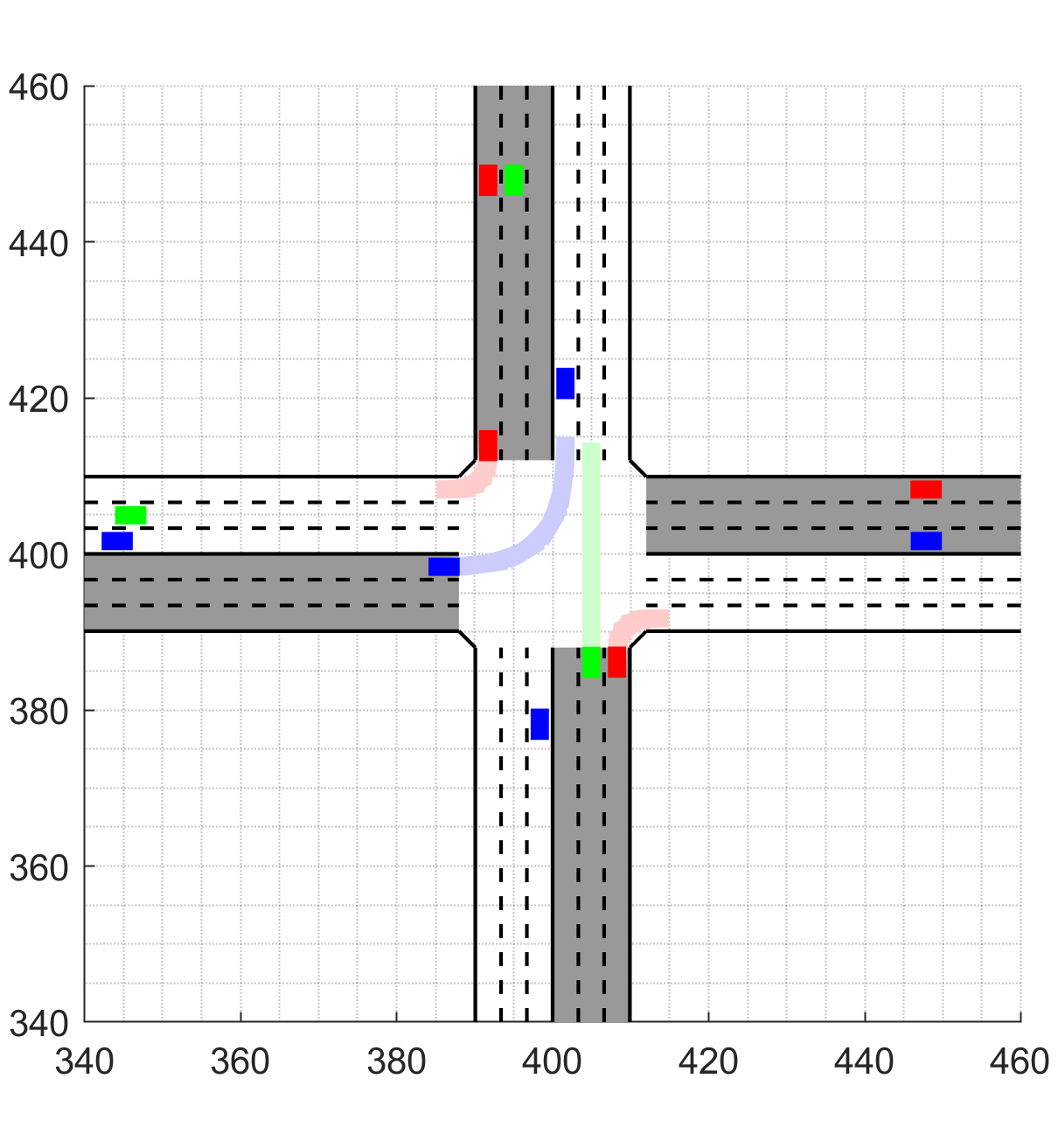}
    \label{2_1610}}
    \subfigure[]{
    \includegraphics[width=0.225\linewidth]{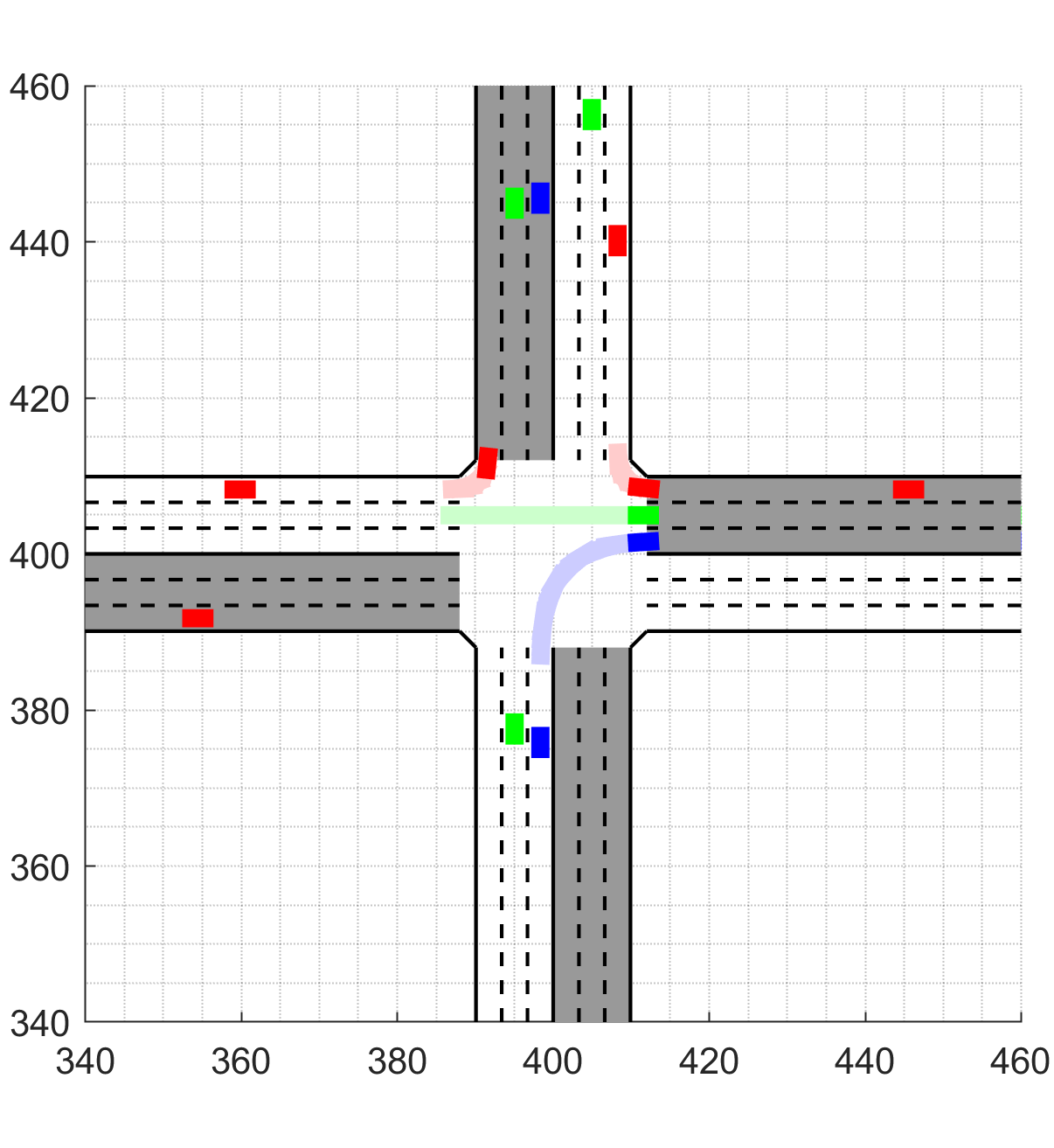}
    \label{2_1680}}
    \subfigure[]{
    \includegraphics[width=0.225\linewidth]{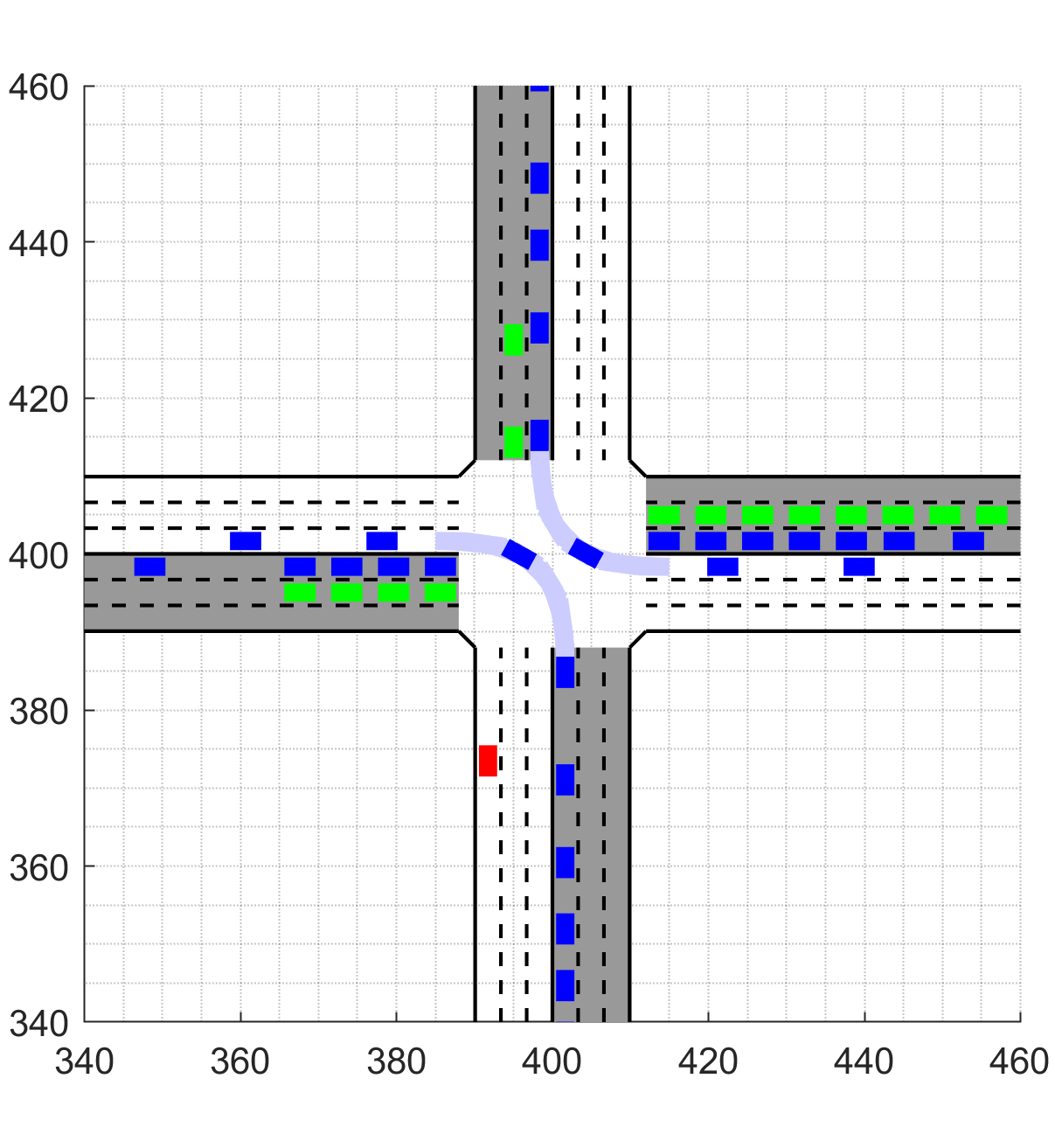}
    \label{3_1210}}
    \subfigure[]{
    \includegraphics[width=0.225\linewidth]{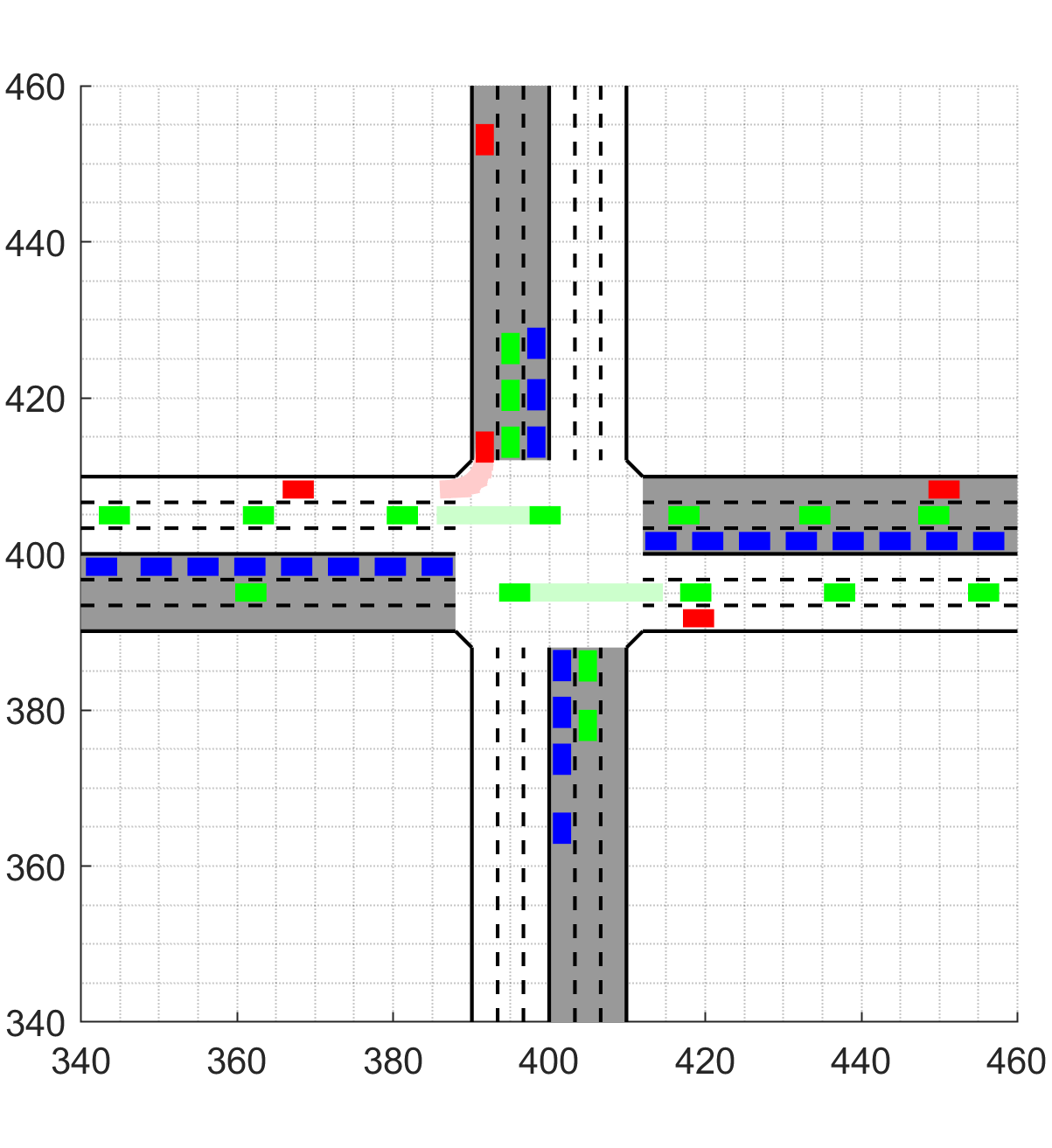}
    \label{3_1530}}
    \subfigure[]{
    \includegraphics[width=0.225\linewidth]{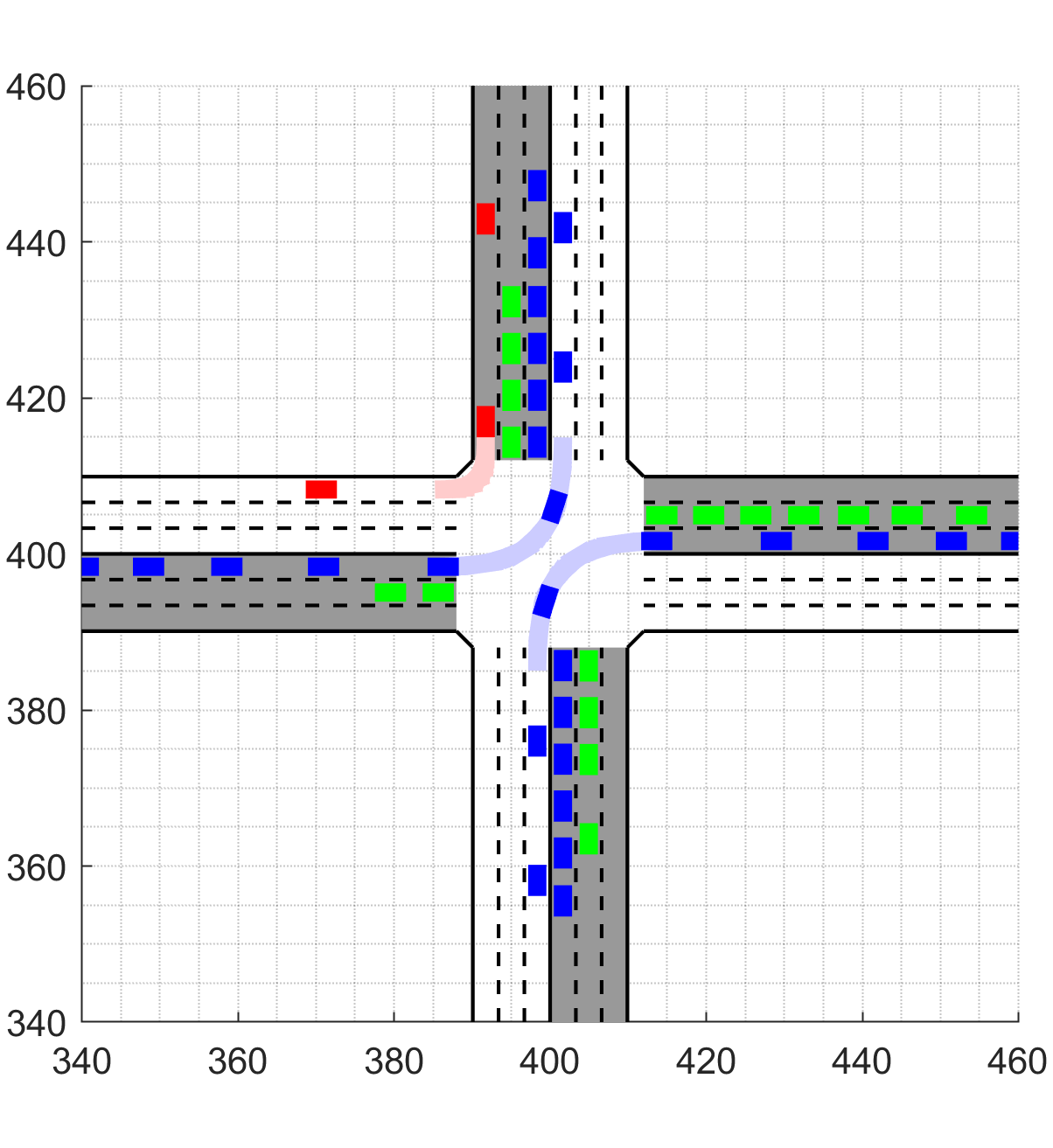}
    \label{3_1690}}
    \subfigure[]{
    \includegraphics[width=0.225\linewidth]{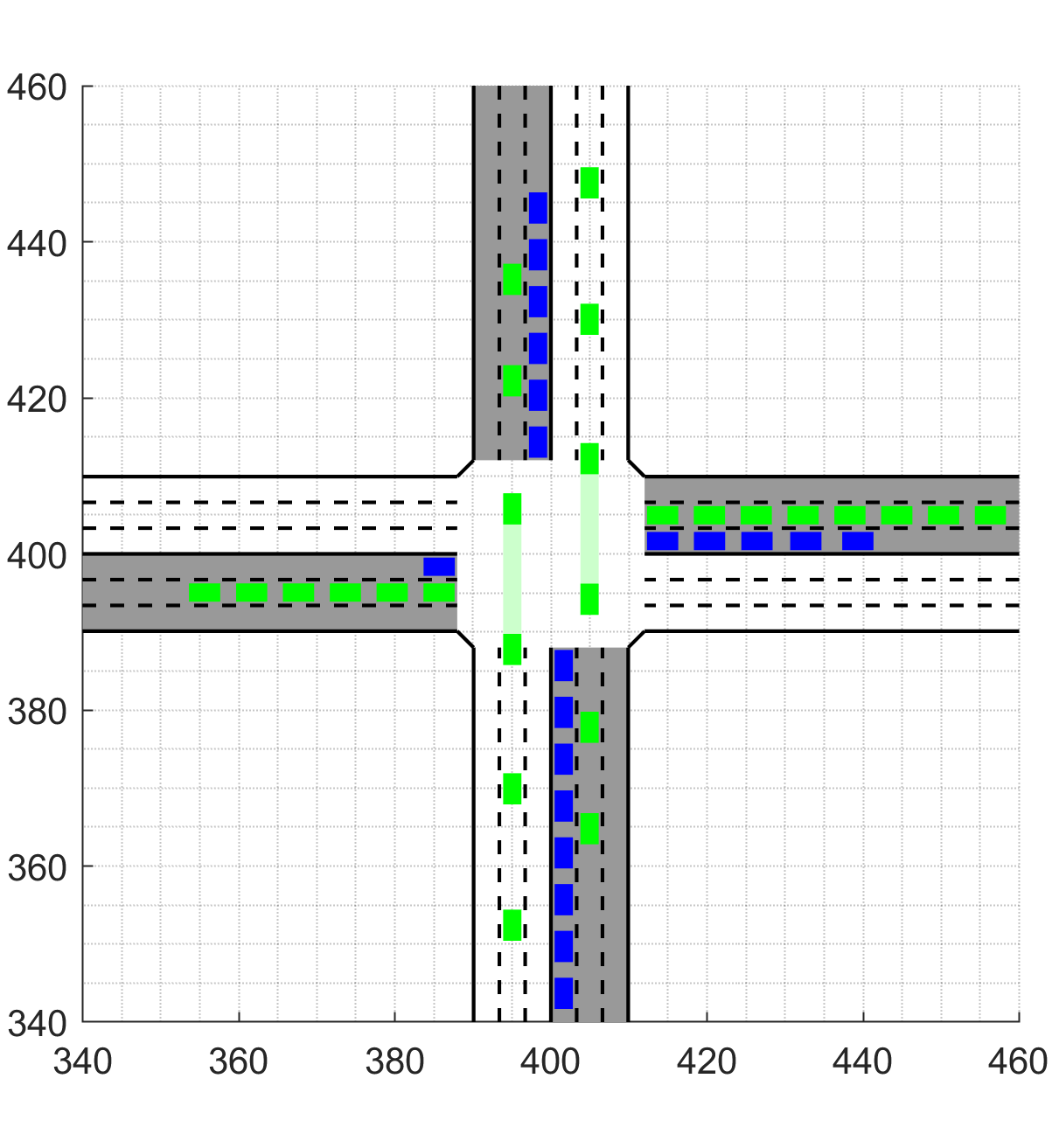}
    \label{3_1910}}
    \caption{Snapshots of conflict zone of the simulation. (a)-(d) are the results of unsignalized intersection with flexible lane direction, (e)-(h) are the results of unsignalized intersection with fixed lane direction, and (i)-(l) are the results of signalized intersection with fixed lane direction. The rectangles represent vehicles and their color shows their turning expectation, where red for right-turning, green for straight-going, and blue for left-turning. The full version video is available online at the address:\protect\\ {\color{blue}https://github.com/cmc623/flexible-lane-direction-intersection}}
    \label{snapshots}
\end{center}
\end{figure}

Snapshots of the west incoming arm at four continuous cycles are presented in Fig.~\ref{segment}. The rectangles represent vehicles and their color shows their turning expectation, where red for right-turning, green for straight-going, and blue for left-turning. The points represent trajectories of vehicles in the former 2.5 seconds. It shows the process of formation reconfiguration where vehicles adjust their position to catch up with their desired layers and lanes. Snapshots at input traffic volume $5000\,\mathrm{vehicle/hour}$ and turning proportion $(0.33,0.33,0.34)$ of the three methods are provided in Fig.~\ref{snapshots} to show the process of simulation. We also show the trajectories of vehicles when passing the conflict zone with lighter color. From the snapshots we can see that our method allows more vehicles to pass the conflict zone simultaneously than the other two methods, and vehicles have to stop and wait for signal timing when using the signalized method.

%
\section{Conclusions}
\label{conc}
%

This paper proposes a multi-lane unsignalized intersection cooperation method that considers flexible lane direction. The two-dimensional distribution of vehicles is calculated and vehicles that are not in conflict are scheduled to pass the intersection simultaneously. The formation reconfiguration method is utilized to achieve cooperative lane changing and longitudinal position adjustment of vehicles. Simulations are conducted at different input traffic volumes and turning proportion of vehicles, and the results indicate that: (1) our method is able to cooperatively control vehicles to avoid collision and pass unsignalized intersections with flexible lane direction, and (2) the performance of our method is better than the unsignalized cooperation method with fixed lane direction and the signalized method.

%
\section*{Acknowledgment}
\label{ack}
This work is supported by the National Key Research and Development Program  of  China  under  Grant  2018YFE0204302,  the  National  Natural Science Foundation of China under Grant 52072212, and Intel Collaborative Research Institute Intelligent and Automated Connected Vehicles. 
%

\bibliography{thesis}

\end{document}